\documentclass[12pt]{article}
\usepackage{amsthm}
\usepackage{graphicx,psfrag,epsf}
\usepackage{enumerate}
\usepackage{url} 
\usepackage{amssymb}
\usepackage{dsfont}
\usepackage{mathtools}
\usepackage{empheq}
\usepackage{adjustbox}
\usepackage{subfigure}
\usepackage{hyperref}
\hypersetup{colorlinks=true,citecolor=blue}
\usepackage{url}
\usepackage{doi}
\usepackage{natbib}
\usepackage{subfigure}
\newtheorem{theorem}{Theorem}[section]

\newtheorem{proposition}{Proposition}[section]
\newtheorem{definition}{Definition}[section]
\newtheorem{remark}{Remark}[section]

\newtheorem{example}{Example}[section]

\usepackage[paperheight=12in,paperwidth=10in]{geometry}

\usepackage{setspace}

\begin{document}
	\begin{center}
		\Large \bf On a Bivariate Copula for Modeling Negative Dependence: Application to New York Air Quality Data
	\end{center}
	\begin{center}
\textbf{ Shyamal Ghosh$^{1}$, Prajamitra Bhuyan$^{2,3}$ and Maxim Finkelstein$^{4,5}$ }\\
$^{1}$  Indian Institute of Information Technology, Guwahati, India\\
		$^{2}$ Queen Mary University London, London, United Kingdom \\
		$^{3}$ The Alan Turing Institute, London, United Kingdom\\
		$^{4}$ University of the Free State, Bloemfontein, South Africa\\
		$^{5}$ University of Strathclyde, Glasgow, United Kingdom

	\end{center}
	
\begin{abstract}
In many practical scenarios, including finance, environmental sciences, system reliability, etc., it is often of interest to study the various notion of negative dependence among the observed variables. A new bivariate copula is proposed for modeling negative dependence between two random variables that complies with most of the popular notions of negative dependence reported in the literature. Specifically, the Spearman's rho and the Kendall's tau for the proposed copula have a simple one-parameter form with negative values in the full range. Some important ordering properties comparing the strength of negative dependence with respect to the parameter involved are considered. Simple examples of the corresponding bivariate distributions with popular marginals are presented. Application of the proposed copula is illustrated using a real data set on air quality in the New York City, USA.   
\end{abstract}

\noindent%
{\it Keywords: Air quality, Inference function for margins, Kolmogorov-Smirnov test, Negatively ordered, Negatively quadrant dependent.}

\section{\large Introduction}\label{intro}

Copulas provide an effective tool for modeling dependence in various multivariate phenomena in the fields of reliability engineering, life sciences, environmental science, economics and finance, etc \citep[Ch-7]{Fontaine_et_al_2020, Cooray_2019, Joe_2006}. Specifically, in recent decades, bivariate copulas were used to generate bivariate distributions with suitable dependence properties \citep{Pos2, Pos3, Finkelstein_2003,  Statistics_2012a, Pos1}. The detailed discussion of historical developments, obtained results and perspectives along with the up to date theory can be found in \cite{Durante_2015} and \cite{Mariu_2018}. It should be noted that most copulas available in the literature possess some limitations in modeling negatively dependent data, which is a certain disadvantage, as negative dependence between vital variables is often encountered in real life. 

\cite{Leh} introduced several concepts of negative dependence for bivariate distributions. Later, \citet{Proschan} and \citet{Yan} extended the corresponding definitions and developed stronger notions of bivariate negative dependence. See \citet{Dist} for detailed discussion on popular dependence notions and their applications in the context of continuous bivariate distributions. \cite{Scarsini_Shaked_1996} provided a detailed overview of the corresponding ordering properties for the multivariate distributions. These results provide useful tools for describing the dependence properties of copulas with respect to a dependence parameter. However, only a few bivariate copulas that allow for a simple and meaningful analysis of this kind have been developed and studied in the literature so far.  The Farlie-Gumbel-Morgenstern (FGM) family of distributions exhibits negative dependence, but the Spearman's rho for this family lies within the interval $[-1/3,1/3]$ \citep{Schu}. 
\citet{Four} and \citet{Three} have considered the four-parameter and the three-parameter extensions of the FGM family proposed by \citet{Smir}, with Spearman's rho lying within the interval $[-0.48, 0.50]$, and $[-0.5, 0.43]$, respectively. To address this issue \citet{New} proposed another extension, but its application is limited because of a singular component that is concentrated on the corresponding diagonal. Some other extensions of the FGM copula are discussed in \citet{FGM2} and \citet{FGM1}. \citet{Hurlimann_2015} have proposed a comprehensive extension of the FGM copula with  the  Spearman's rho and Kendall's tau attaining any value in $(-1,1)$. However, the dependence properties and ordering properties of these  copulas are not well studied in the literature. Recently, \cite{Cooray_2019} proposed a new extension of FGM family which exhibits negative dependence among the variables in a very strong sense. However, its Spearman's rho and Kendll's tau are restricted to $[-0.70,0]$, and $[-0.52,0]$, respectively.

Thus, it is quite a challenging problem to construct a flexible bivariate copula with the correlation coefficient that takes any value in the interval $(-1,0)$. Moreover, it is not sufficient just to suggest this type of copula, but it is essential to describe its properties (including relevant stochastic comparisons) especially in the case of strong notions of dependence. In many real life scenarios, paired observations of non-negative variables possess strong negative dependence. For example, rainfall intensity and duration are jointly modeled incorporating their negative dependence for the study of derived flood frequency distribution \citep{Rain}. This paper is motivated by a real case study on air quality for New York Metropolitan area where the joint distribution of the wind speed and ozone level exhibits strong negative dependence (See Section \ref{Case}). We believe that the current study meets to some extent this challenge, as we propose an absolutely continuous negatively dependent copula that satisfies most of the popular notions of negative dependence available in the literature with correlation coefficients in the interval $(-1,0)$. 

The paper is organized as follows. In Section \ref{Copu}, we describe the baseline (for the proposed copula) distribution and discuss some basic properties including conditional distributions and correlation coefficients. Various notions of negative dependence in the context of the proposed copula and ordering properties are considered in Section \ref{NProp}, and Section \ref{OrderP}, respectively. Section \ref{Exam}, provides some examples of negatively dependent standard bivariate distributions. The estimation methodologies are discussed in Section \ref{EST}. In Section \ref{Case}, as an illustration,  we provide a real case study. Finally, some concluding remarks are given in Section \ref{Con}.

\section{The Bivariate Copula}\label{Copu}

\cite{Bhuyan_2020} proposed a negatively dependent bivariate life distribution that possesses nice closed-form expressions for the joint distributions and exhibits various strong notions of negative dependence reported in the literature. Most importantly, the correlation coefficient may take any value in the interval $(-1,0)$. One of the marginal distribution is Exponential and the other belongs to skew log Laplace family \citep{Dixit_Khandeparkar_2017}. We utilize the negative dependence structure inherent in this model and formulate a copula with strong negative dependence. The joint distribution function and the marginal distributions are given by
\begin{equation}\label{Joint}
	H(x,y) = \begin{dcases}
		y^{\lambda} - e^{-\lambda x}+\dfrac{\lambda}{(\lambda+\mu)y^{\mu}}\left[ e^{(\lambda+\mu)x}-y^{\lambda+\mu}\right], & 0<y \leq 1, x>-\log y \\
		1-e^{-\lambda x}-\dfrac{\lambda}{(\lambda+\mu)y^{\mu}}\left[ 1-e^{-(\lambda+\mu)x}\right], & x>0, y>1, 
	\end{dcases}
\end{equation}
and  
$F(x)= 1- e^{-\lambda x}$ for $x>0$, and $G(y) =\dfrac{\mu}{(\lambda+\mu)}y^{\lambda}\mathds{1}(0<y \leq1)+ \left[1-\dfrac{\lambda}{(\lambda+\mu)y^{\mu}}\right]\mathds{1}( y>1)$, respectively, where $\lambda,\mu >0$.
Note that $F(\cdot)$ and $G(\cdot)$ are continuous. We first find the quasi-inverse functions of $F(\cdot)$ and $G(\cdot)$ and `put' those into the arguments of the joint distribution function $H(\cdot,\cdot)$ given by (\ref{Joint}).
Then by Corollary 2.3.7 of \citet[p-22]{Nelson_2006}, we obtain the following copula
\begin{equation}\label{Cop}
	C_{\lambda, \mu}(u,v) = \begin{dcases}
		v-(1-u)+\dfrac{\lambda\mu^{\frac{\mu}{\lambda}}}{(\lambda+\mu)^{1+\frac{\mu}{\lambda}}}(1-u)^{{1+\frac{\mu}{\lambda}}}v^{-\frac{\mu}{\lambda}},\\
		\hspace{100pt} 0<v \leq\dfrac{\mu}{\mu+\lambda}, 1-\dfrac{(\lambda +\mu)v}{\mu}<u<1, \\
		u-(1-v)\left[1-(1-u)^{{1+\frac{\mu}{\lambda}}}\right], \\ \hspace{100pt}  0<u<1, \dfrac{\mu}{\mu+\lambda}<v<1.
	\end{dcases}
\end{equation}
Now using the reparameterization $\mu=\theta \lambda $, in (\ref{Cop}), we rewrite $C_{\lambda, \mu}$ as 
\begin{equation}\label{Copula}
	C_\theta(u,v) = \begin{dcases}
		v-(1-u)+\dfrac{\theta^\theta}{(1+\theta)^{1+\theta}}(1-u)^{1+\theta}v^{-\theta}, \\\hspace{100pt} 0<v \leq\dfrac{\theta}{1+\theta}, 1-\dfrac{(1+\theta)v}{\theta}<u<1 \\
		u-(1-v)\left[1-(1-u)^{1+\theta}\right], \\ \hspace{100pt} 0<u<1, \dfrac{\theta}{1+\theta}<v<1, 
	\end{dcases}
\end{equation}
for $\theta >0$.
It is easy to verify that $C_\theta(u,v)$, given by (\ref{Copula}), satisfies the following conditions: (i) $C_\theta(u,0)=0=C_\theta(0,v)$, (ii) $C_\theta(u,1)=u$, $C_\theta(1,v)=v$, for any $u$, $v$ in $I=[0,1]$, and (iii) $C_\theta(u_2,v_2)-C_\theta(u_2, v_1)-C_\theta(u_1, v_2)+C_\theta(u_1,v_1)\geq 0$, for any $u_1, u_2, v_1, v_2$ in $I$ with $u_1\leq u_2$ and $v_1\leq v_2$. In Figure \ref{Copula_plot}-\ref{Cont_plot}, we provide graphical presentation of the proposed copula for different values of the dependence parameter $\theta$.
\begin{figure}[htp]
	\centering
	\subfigure [Copula plot for $\theta=0.1$] {
		\includegraphics[scale=0.5]{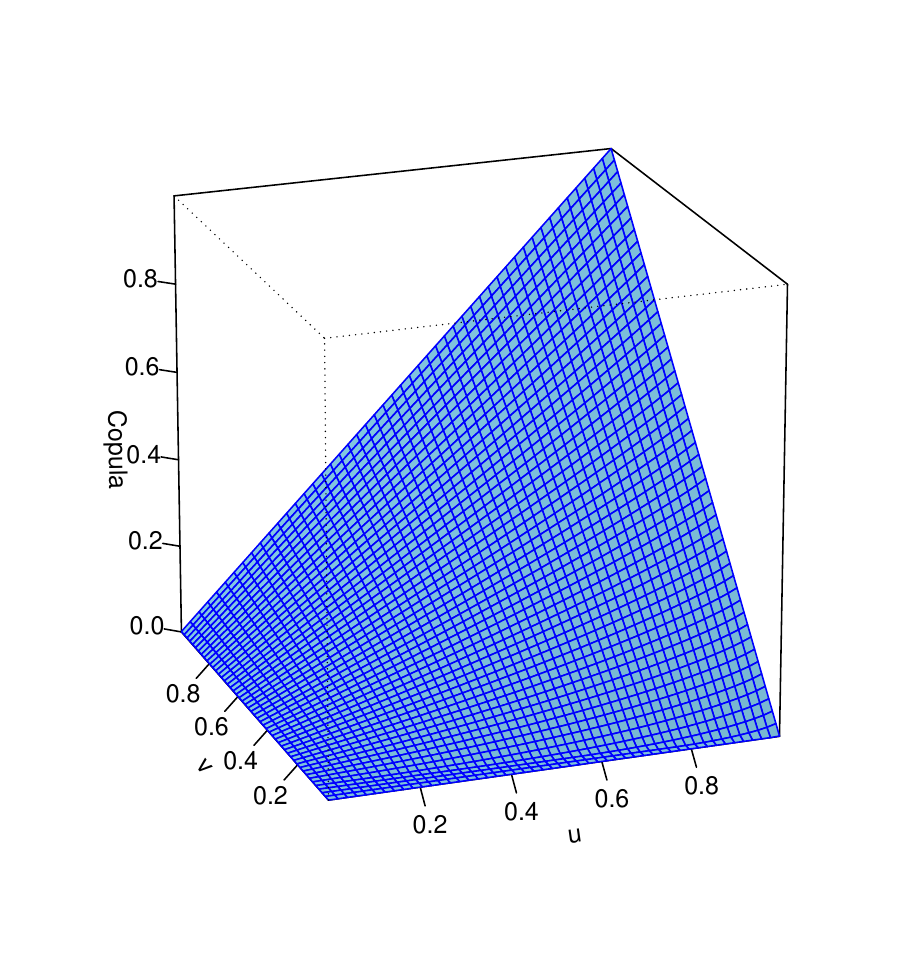}
		\label{fig:subfig01}
	}
	\subfigure[Copula plot for $\theta=1$]{
		\includegraphics[scale=0.5]{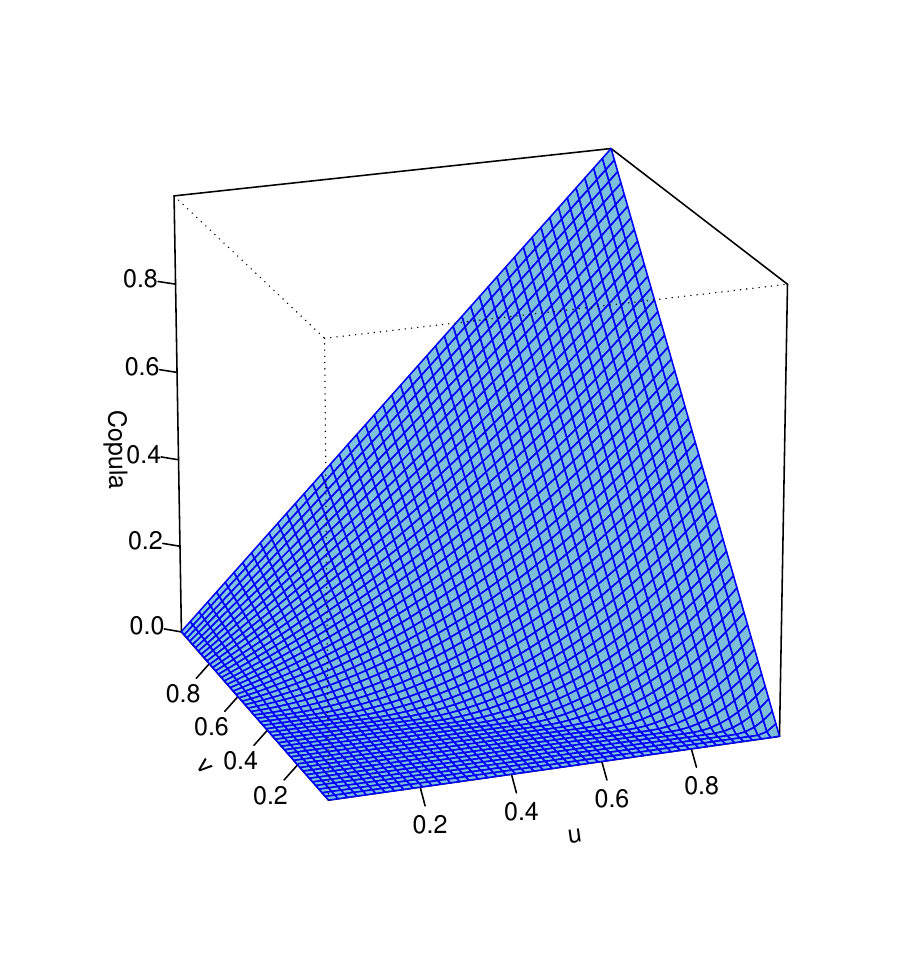}
		\label{fig:subfig02}
	}
	\subfigure[Copula plot for $\theta=5$]{
		\includegraphics[scale=0.5]{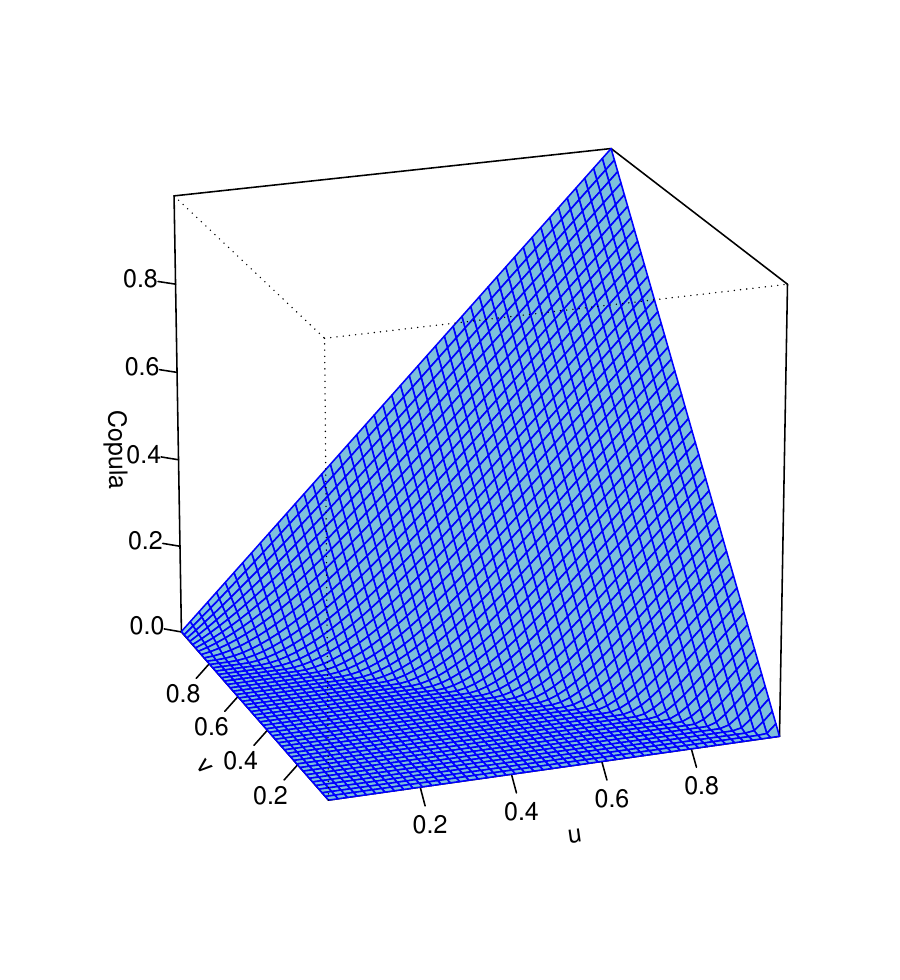}
		
	}
	\subfigure[Copula plot for $\theta=10$]{
		\includegraphics[scale=0.5]{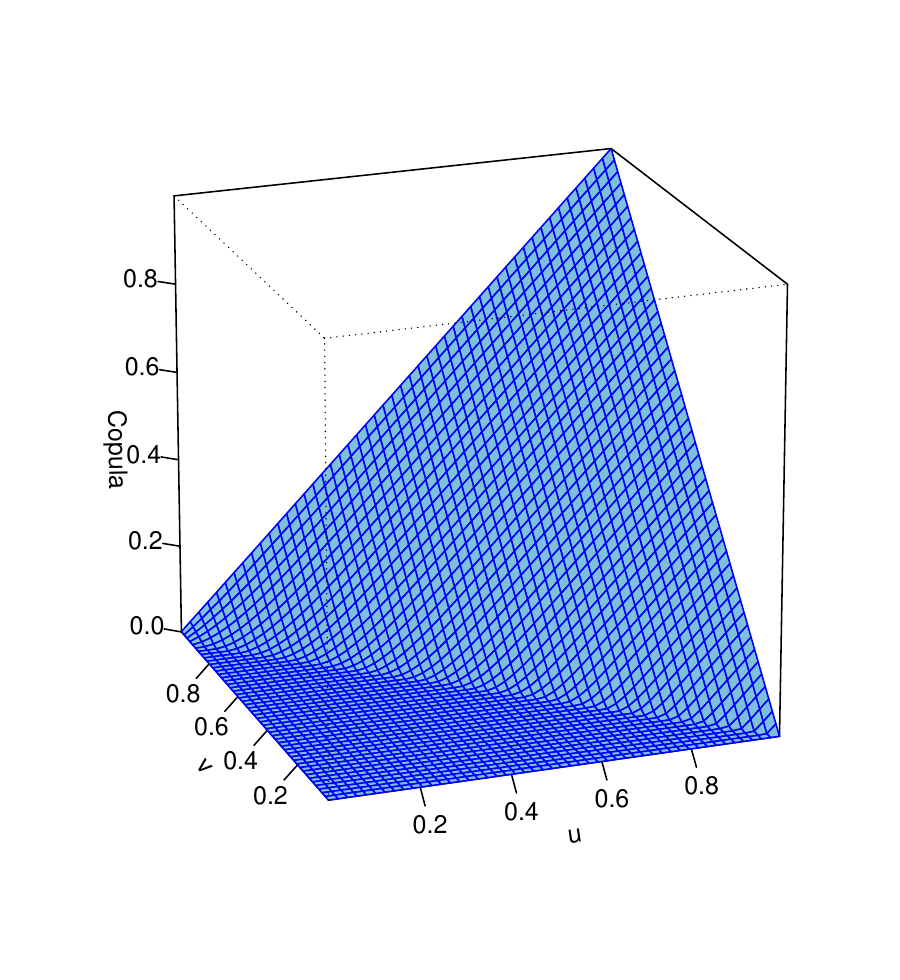}
		\label{fig:subfig03}
	}
	\caption{ Graphical plots of $C_{\theta}$ for different choices of $\theta$ on an unit square.}
	\label{Copula_plot}
\end{figure}

\begin{figure}[htp]
	\centering
	\subfigure [Contour plot for $\theta=0.1$] {
		\includegraphics[scale=0.5]{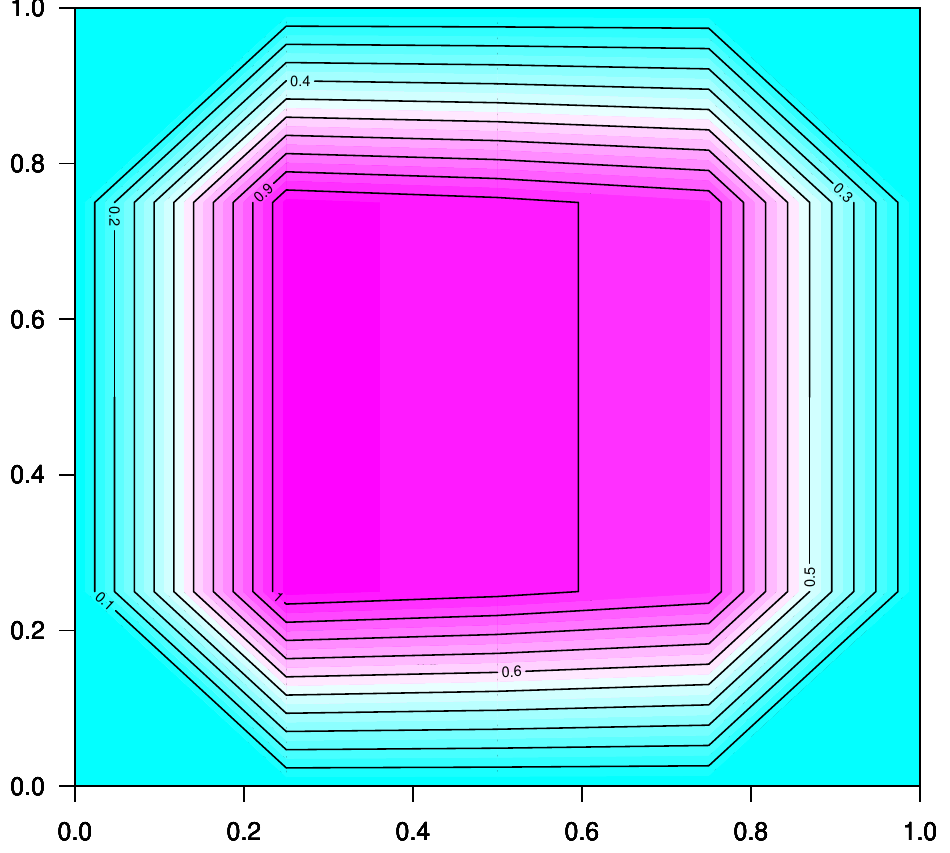}
		\label{fig:subfig1}
	}
	\subfigure[Contour plot for $\theta=1$]{
		\includegraphics[scale=0.5]{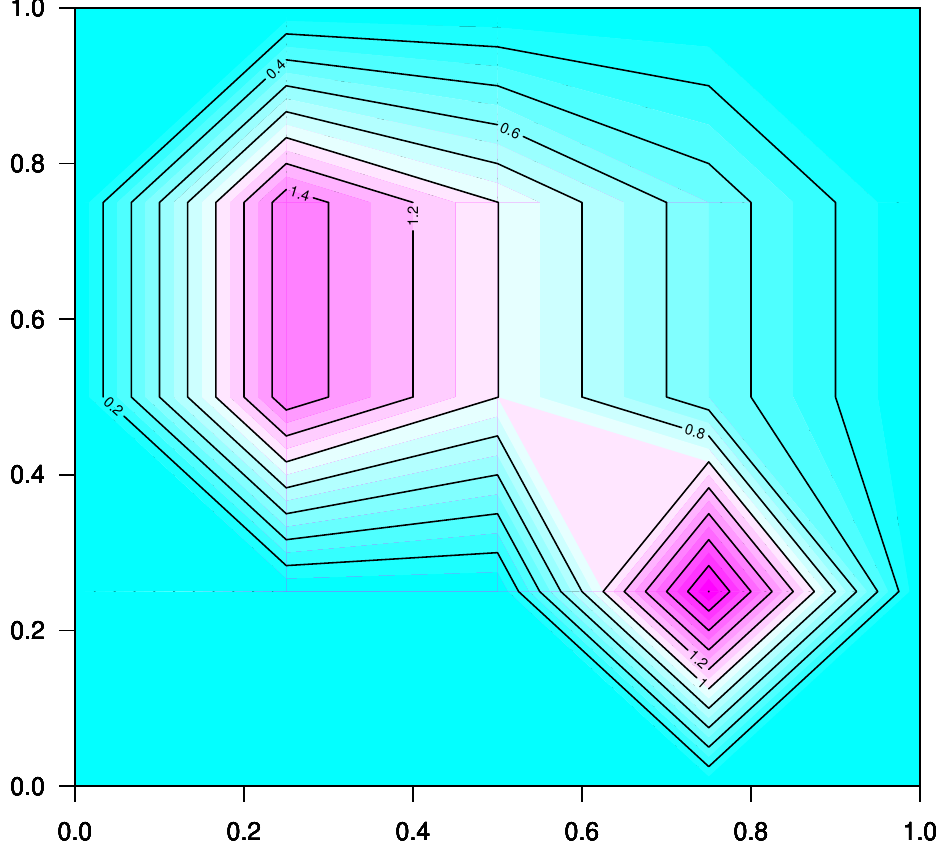}
		\label{fig:subfig2}
	}
	\subfigure[Contour plot for $\theta=5$]{
		\includegraphics[scale=0.5]{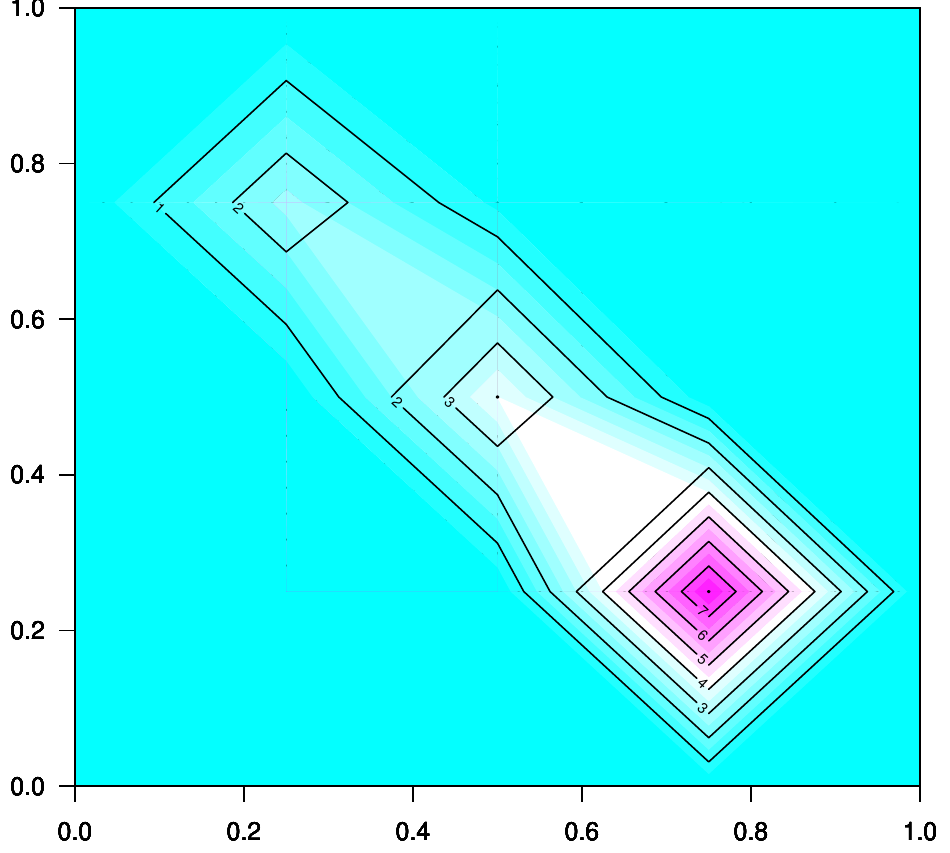}
		
	}
	\subfigure[Contour plot for $\theta=10$]{
		\includegraphics[scale=0.5]{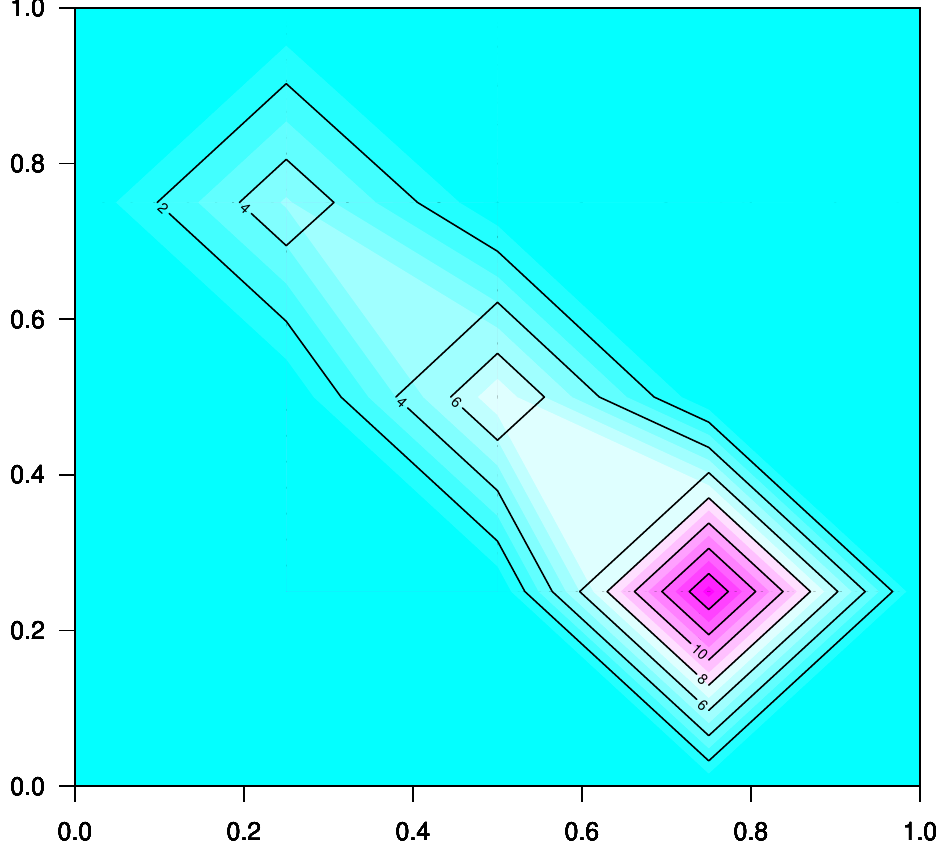}
		\label{fig:subfig3}
	}
	\caption{Contour plots of $C_{\theta}$ for different choices of $\theta$.}
	\label{Cont_plot}
\end{figure}

The survival copula, is the function $\bar{C}$ which couples the joint survival function to its marginal survival functions. It is easy to show that $\bar{C}$ is a copula, and is related to the copula $C$ via the equation $\bar{C}= u+v-1 + C(1-u, 1-v)$. See \citet[p-32]{Nelson_2006} for details. The survival copula and the density function of the proposed copula $C_\theta(u,v)$ are given by 
\begin{equation}
	\bar{C}_\theta(u,v) = \begin{dcases}
		\dfrac{\theta^\theta}{(1+\theta)^{1+\theta}}u^{1+\theta}(1-v)^{-\theta}, &  \dfrac{1}{1+\theta}\leq v < 1, 0<u<\dfrac{(1+\theta)}{\theta}(1-v) \nonumber\\
		vu^{(1+\theta)},  & 0<u<1, 0<v<\dfrac{1}{1+\theta}. 
	\end{dcases}
\end{equation}
and
\begin{equation}\label{eqn_den}
	c_\theta(u,v) = \begin{dcases}
		\dfrac{\theta^{1+\theta}}{(1+\theta)^\theta}(1-u)^\theta v^{-(1+\theta)}, & 0<v \leq\dfrac{\theta}{1+\theta}, 1-\dfrac{(1+\theta)v}{\theta}<u<1 \\
		(1+\theta)(1-u)^\theta, & 0<u<1, \dfrac{\theta}{1+\theta}<v<1,
	\end{dcases}
\end{equation}
respectively.

\subsection{Conditional Copulas}
The conditional copula of $U$ given $V=v$, is as follows.
For $0<v \leq\frac{\theta}{(1+\theta)}$,
\begin{equation}\label{U_V_1}
	C_\theta (u\mid v)= 
	1- \frac{\theta^{(1+\theta)}}{(1+\theta)^{(1+\theta)}}(1-u)^{(1+\theta)} v^{-(1+\theta)},  \,\,\,\,1-\dfrac{(1+\theta)v}{\theta}<u<1, \end{equation}
whereas for $\frac{\theta}{(1+\theta)}<v<1$,
\begin{equation}\label{U_V_2}
	C_\theta(u\mid v)=  
	1-(1-u)^{(1+\theta)},
	\,\,\,\,\,0<u<1.
\end{equation}

The conditional mean and variance of $U\mid V=v$ are given by
\begin{equation}
	E[U\mid V=v]= \begin{dcases}
		1-\dfrac{(1+\theta)^2v}{\theta(\theta+2)}, & 0<v \leq\dfrac{\theta}{1+\theta} \nonumber\\
		\dfrac{1}{\theta+2}, &  
		\dfrac{\theta}{1+\theta}<v<1, 
	\end{dcases}
\end{equation}
and
\begin{equation}
	Var[U\mid V=v]= \begin{dcases}
		\dfrac{(1+\theta)^3v^2}{\theta^2(\theta +2)^2(\theta +3)},& 0<v \leq\dfrac{\theta}{1+\theta} \nonumber\\
		\dfrac{\theta+1}{(\theta+2)^2(\theta+3)},& \dfrac{\theta}{1+\theta}<v<1, 
	\end{dcases}
\end{equation}
respectively.

\begin{remark}
	{\rm The regression of $U$ on $V=v$ is linearly decreasing in $v$ for $0<v\leq \frac{\theta}{\theta+1}$, and independent of $v$ for $\frac{\theta}{\theta+1}<v<1$. Also, it is interesting to note that the conditional variance of $U\mid V=v$ is an increasing function of $v$ and bounded from above by $\frac{\theta+1}{(\theta+2)^2(\theta+3)}.$}
\end{remark}

The conditional copula of $V$ given $U=u$, is given by
\begin{equation}\label{V_U}
	C_\theta(v\mid u)= \begin{dcases}
		1- \dfrac{\theta^\theta}{(1+\theta)^\theta}(1-u)^\theta v^{-\theta},  & \dfrac{(1-u)\theta}{(1+\theta)}<v \leq\dfrac{\theta}{1+\theta} \\
		1-(1+\theta)(1-v)(1-u)^\theta, &  \dfrac{\theta}{1+\theta}<v<1 
	\end{dcases}
\end{equation}
The conditional mean  and variance of $V\mid U=u$, are given by  \begin{equation}E[V\mid U=u]=\dfrac{(1-u)^\theta}{2(1-\theta)}-\dfrac{2\theta^2(1-u)}{1-\theta^2}\nonumber,
\end{equation}
for $\theta \neq 1$, and 
\begin{eqnarray*}
	Var[V\mid U=u]&=&-\dfrac{(1+\theta)(1-u)^\theta\left[ 2 - \theta + \theta^2 (2 - 6 u) + 3 \theta^3 u\right]}{3(\theta-2)(\theta^2-1)^2}+\dfrac{\theta^3(1-u)^2}{(\theta-2)(\theta^2-1)^2}\\
	&-&\dfrac{(1+\theta)^2(1-u)^{2\theta}}{4(\theta^2-1)^2}
\end{eqnarray*}
for $\theta \neq 1,2$, respectively. 
\begin{remark}
	{\rm The regression of $V$ on $U=u$ is strictly decreasing in $u$.}  
\end{remark}

One can use the conditional copula of $U$ given $V=v$, provided in (\ref{U_V_1}) and (\ref{U_V_2}), to simulate from the proposed copula $C_{\theta}$, given by (\ref{Copula}), using the following steps.
\begin{itemize}
	\item[Step I.] Simulate $v_i$ and $u_i^{\ast}$ independently from standard uniform distribution.
	\item[Step II.] If $v_i\le \frac{\theta}{\theta+1}$, then solving $C_{\theta}(u\mid v_i)=u_i^{\ast}$ from (\ref{U_V_1}), we get $u_i=1-(\frac{\theta+1}{\theta})v_i(1-u_i^{\ast})^{\frac{1}{1+\theta}}$;\\
	else, solving $C_{\theta}(u\mid v_i)=u_i^{\ast}$ from (\ref{U_V_2}), we get $u_i=1-(1-u_i^{\ast})^{\frac{1}{1+\theta}}$.
	\item[Step III.] Repeat Step I and Step II $n$ times to obtain independently and identically distributed realizations $(u_i,v_i)$, for $i=1,2,\ldots,n$  from $C_{\theta}$.
\end{itemize}
A similar algorithm can be elaborated to simulate from $C_{\theta}$ based on the conditional copula of $V$ given $U$, provided in (\ref{V_U}).The associated R programme for the aforementioned algorithm are provided in the Supplementary material. The Scatter plots based on $500$ simulated observations using the aforementioned algorithm for four different values of $\theta$ are given in Figure \ref{Scatter1}. As expected, the data points are getting closer to the diagonal $v=-u$ for higher values of $\theta$.

\begin{figure}[htp]
	\centering
	\subfigure [Scatter plot for $\theta=0.1$] {
		\includegraphics[scale=0.5]{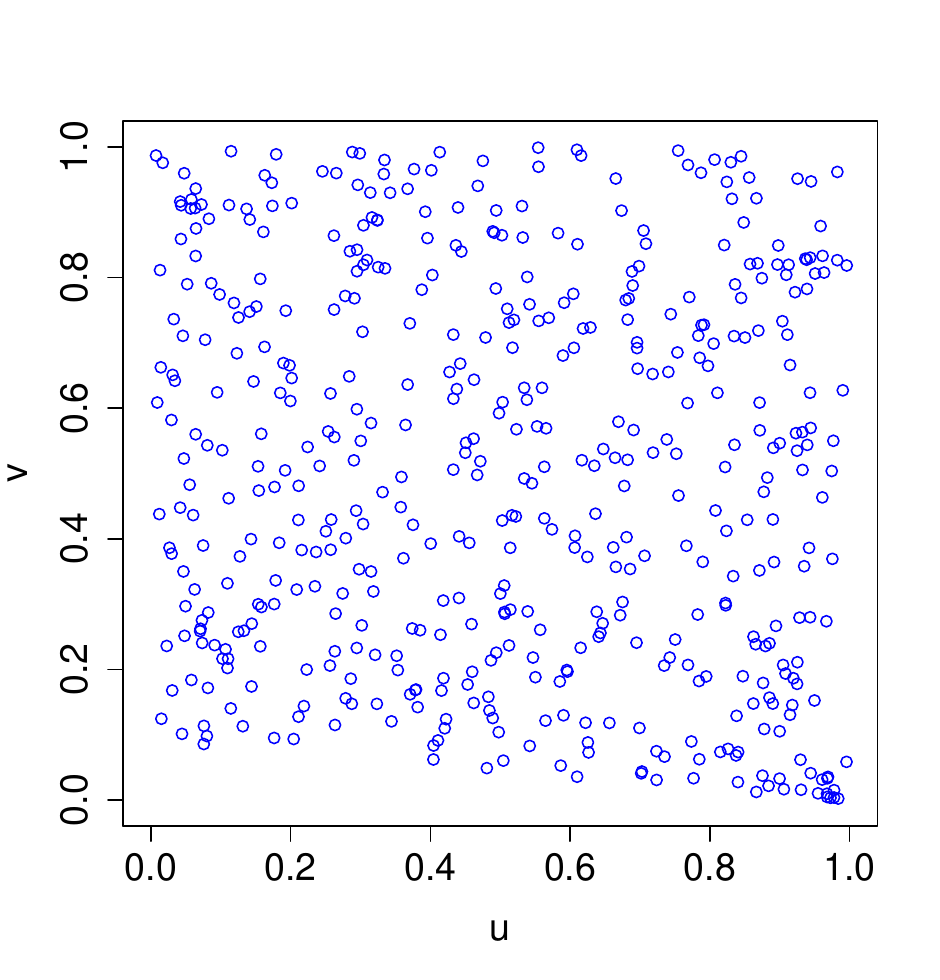}
		\label{fig:subfig001}
	}
	\subfigure[Scatter plot for $\theta=1$]{
		\includegraphics[scale=0.5]{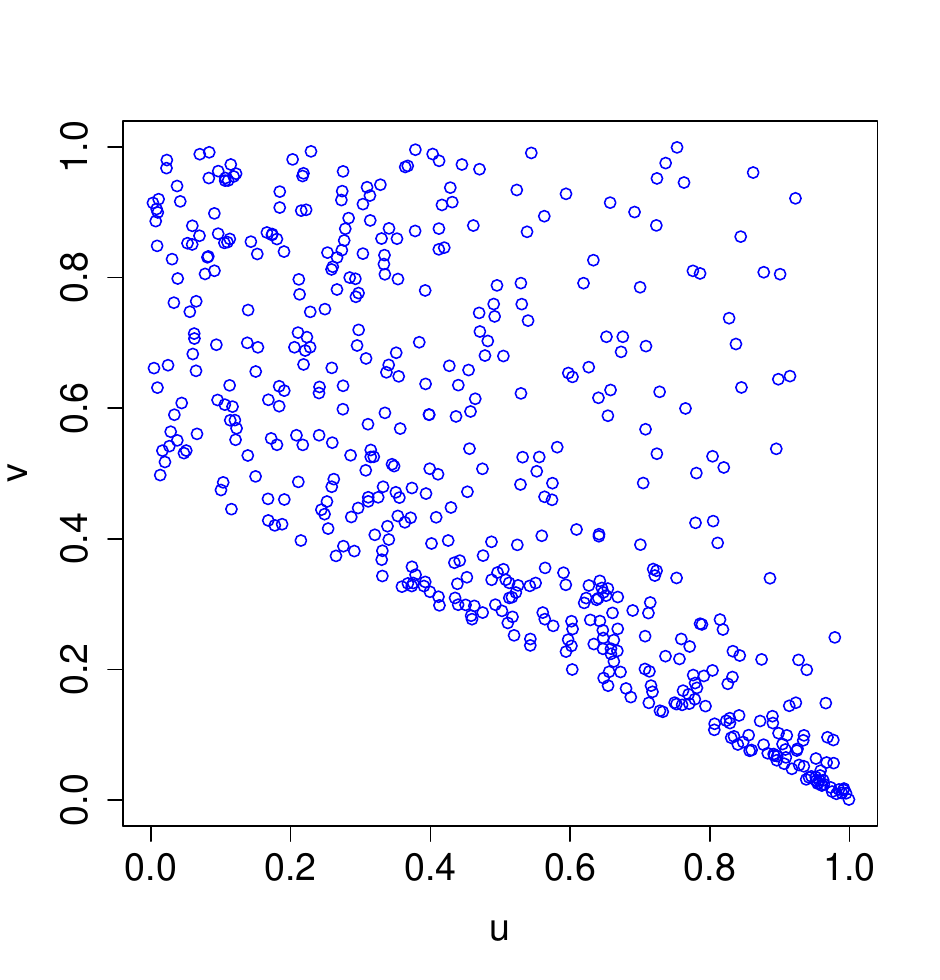}
		\label{fig:subfig002}
	}
	\subfigure[Scatter plot for $\theta=5$]{
		\includegraphics[scale=0.5]{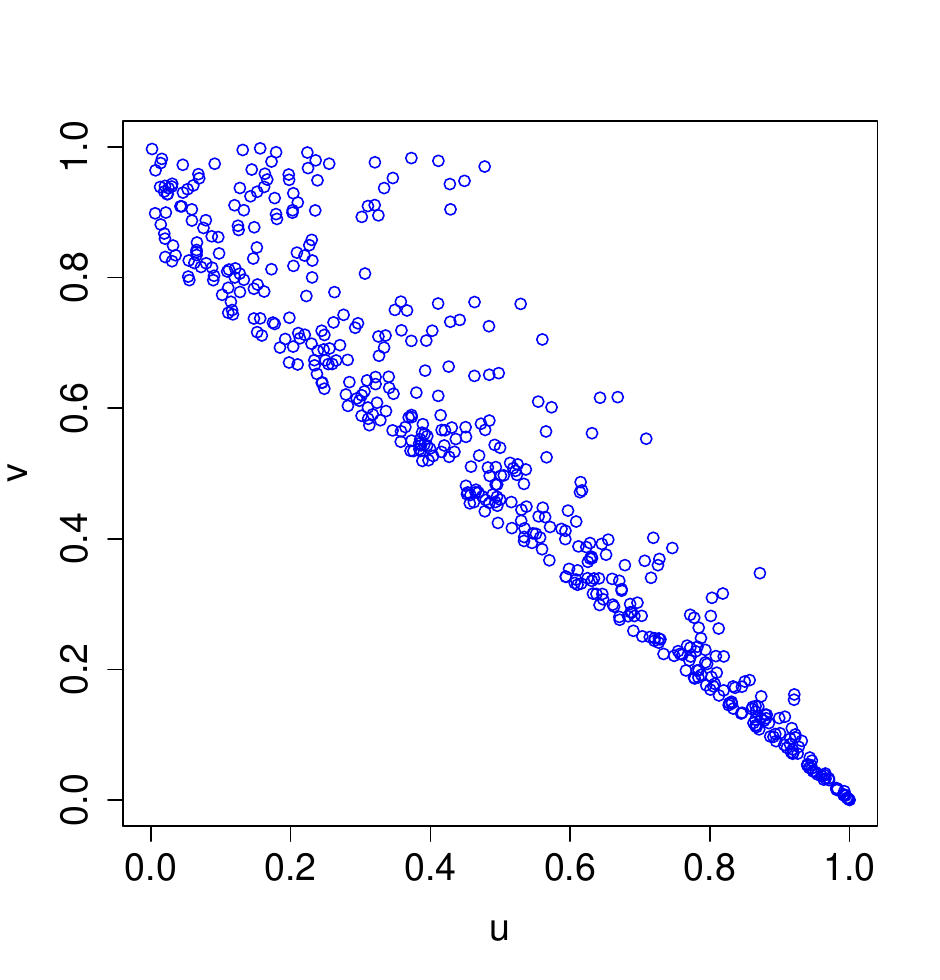}
		
	}
	\subfigure[Scatter plot for $\theta=10$]{
		\includegraphics[scale=0.5]{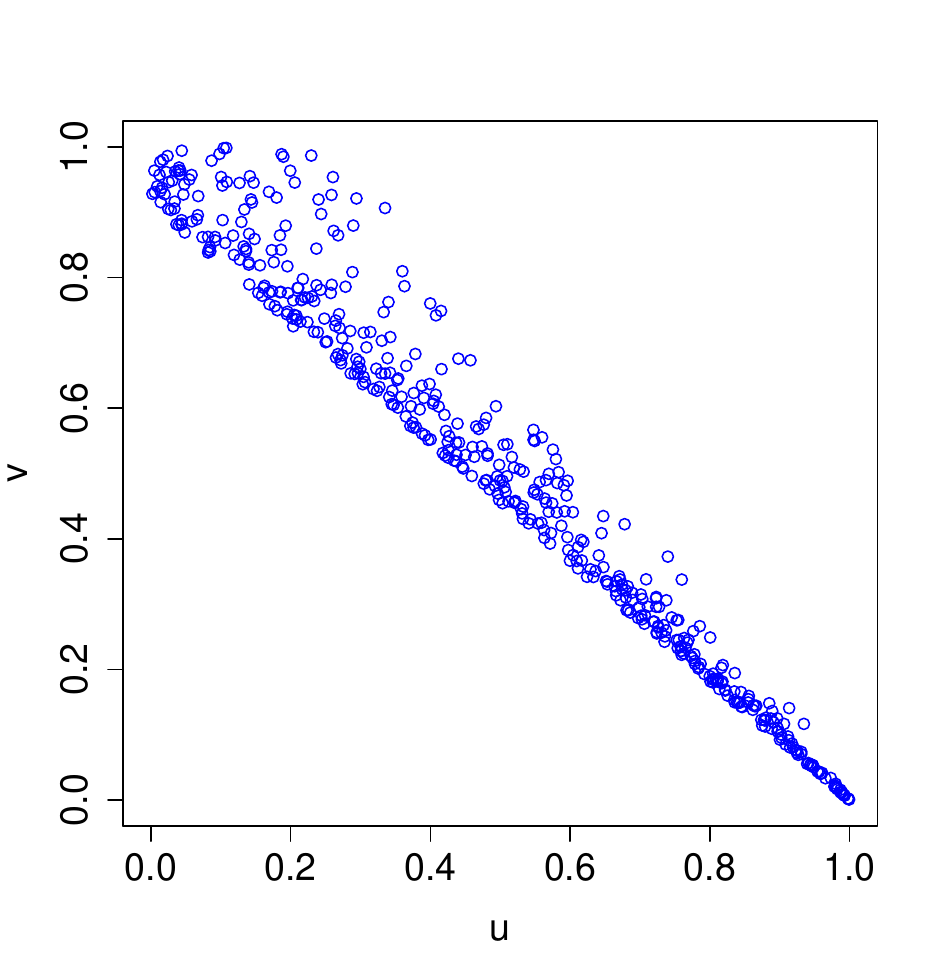}
		\label{fig:subfig003}
	}
	\caption{ Scatter plots based on $500$ simulated observations from $C_{\theta}$ for different choices of $\theta$.}
	\label{Scatter1}
\end{figure}

\subsection{Basic Properties}\label{Prop}
In this Subsection, we present three important propositions related to the proposed copula. The detailed proofs are presented in Appendix A. 
\begin{proposition}\label{pro21}
	The copula $C_{\theta}$, defined in (\ref{Copula}), is  decreasing with respect to its dependence parameter $\theta$, i.e., if $\theta_1\leq \theta_2$ then $C_{\theta_2}(u,v)\leq C_{\theta_1}(u,v)$, for all $(u,v) \in I^2 = [0,1]\times[0,1]$. 
\end{proposition}

\begin{proposition}\label{pro22}
	The copula $C_{\theta}$, defined in (\ref{Copula}), is sub-harmonic, i.e., $\nabla^2 C_\theta(u,v) \geq 0$. 
\end{proposition}

\begin{proposition}\label{pro23}
	The copula $C_{\theta}$, defined in (\ref{Copula}), is absolutely continuous.
\end{proposition}

\subsection{Measures of Dependence }\label{Dependence}
Measures of dependence are commonly used to summarize the complicated dependence
structure of bivariate distributions. See \citet[Ch-2]{Joe_1997}, \citet[Ch-5]{Nelson_2006} and \citet[Ch-2]{Mariu_2018} for a detailed review on measures of dependence and its associated properties. In this section, we derive the expressions of the Kendall’s tau and the Spearman’s rho for the proposed copula $C_\theta$. Essentially, these coefficients measure the correlation between the ranks rather than actual values of $X$ and $Y$. Therefore, these coefficients are unaffected by any monotonically increasing transformation of $X$ and $Y$.
\begin{definition}
	Let $X$ and $Y$ be the continuous random variables with the dependence structure described by the copula $C$. Then the population version
	of the Spearman’s rho for $X$ and $Y$ is given by
	$$\rho :=\int_0^1\int_0^1uvdC(u,v)-3= \int_0^1\int_0^1C(u,v)dudv-3$$
\end{definition}
\begin{proposition}
	Let $(X, Y)$ be a random pair with copula $C_\theta$. The Spearman’s rho
	is given by $$\rho =\dfrac{2(3+3\theta+\theta^2)}{2+3\theta+\theta^2}-3,$$ which is a decreasing function in $\theta$ and takes any values between -1 and 0. 
\end{proposition}
\begin{definition}
	Let $X$ and $Y$ be the continuous random variables with copula $C$. Then, the population version of the Kendall’s tau for $X$ and $Y$
	is given by
	$$\tau := 4\int_0^1\int_0^1C(u,v)dC(u,v)-1$$
\end{definition}
\begin{proposition}
	Let $(X, Y)$ be a random pair with copula $C_\theta$. Then the Kendall’s tau
	is given by $$\tau =\dfrac{-\theta}{(1+\theta)},$$ which is a decreasing function in $\theta$ and takes any values between -1 and 0. 
\end{proposition}
In Figure \ref{Corr_plot}, we have plotted the Spearman’s rho and the Kendall’s tau against the dependence parameter $\theta$. It is easy to see that the Spearman’s rho is less than the Kendall’s tau for all $\theta>0$.



\begin{figure}
	\centering
	\includegraphics[scale=0.65]{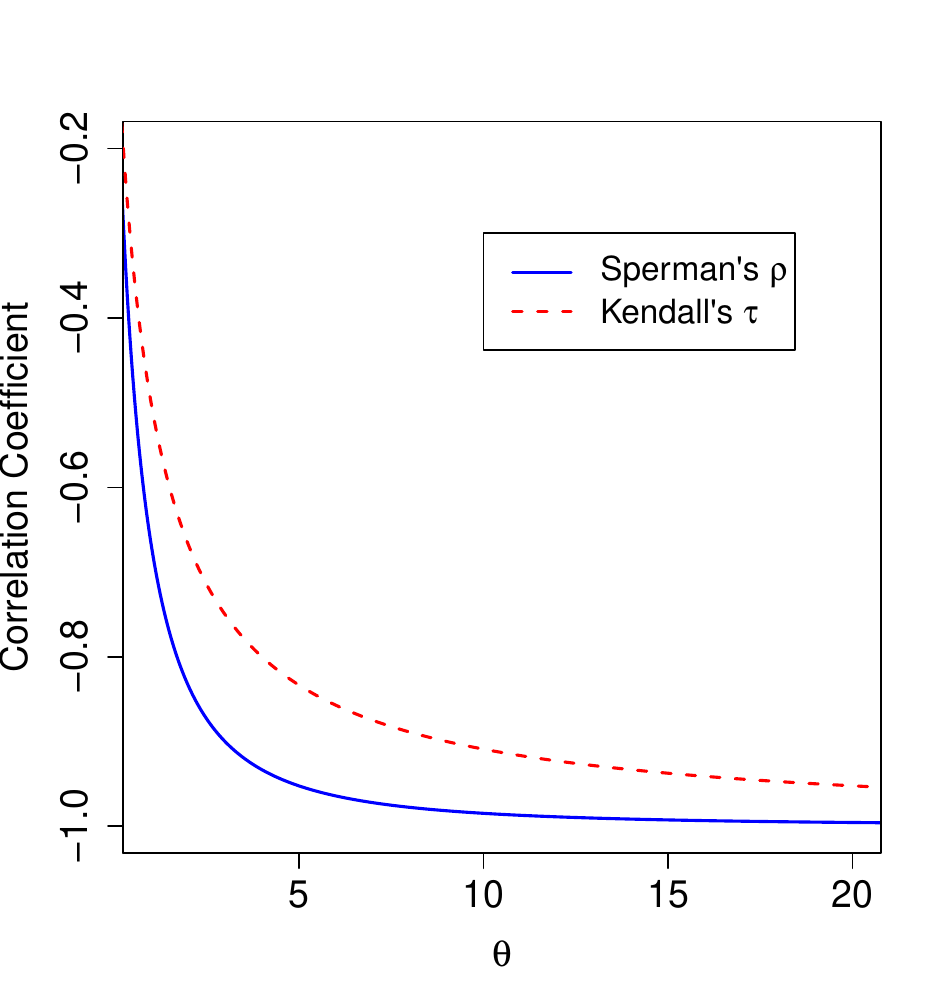}
	\caption{Plot of Spearman's rho and Kendall's tau against the dependence parameter $\theta$.}
	\label{Corr_plot}
\end{figure}

\section{Connections with notions of Negative Dependence}\label{NProp}
As discussed in Subsection \ref{Dependence}, the Spearman's rho and the Kendall's tau measure the correlation between two random variables. However, it is possible that these random variables may have the strong correlation, but possess the weak association with respect to different notions of dependence or vice versa. In this section, we discuss several relevant notions of negative dependence, namely \textit{Quadrant Dependence}, \textit{Regression Dependence} and \textit{Likelihood Ratio Dependence}, etc., and explore, whether the corresponding properties are satisfied by the proposed copula or not. First, we provide the definitions of the aforementioned dependence notions as discussed in 
\cite{Nelson_2006} and \cite{Dist}.
\begin{definition}\label{Orders}
	Let $X$ and $Y$ be continuous random variables with copula $C$. Then
	\begin{enumerate}
		\item $X$ and $Y$ are Negatively Quadrant Dependent (NQD) if $P(X \leq x, Y \leq y) \leq {P(X \leq x)P(Y \leq y)}$, for all $(x, y) \in R^2$, where $R^2$ is the domain of joint distribution of $X$ and $Y$, or equivalently a copula $C$ is said to be NQD if for all $(u,v)\in I^2$, $C(u,v)\leq uv.$
		\item $Y$ is left tail increasing in $X$ (LTI($Y\mid X$)), if $P[Y \leq y \mid X \leq x]$ is a nondecreasing function of $x$ for all $y$.
		\item $X$ is left tail increasing in $Y$ (LTI($X\mid Y$)), if $P[X \leq x \mid Y \leq y]$ is a nondecreasing function of $y$ for all $x$.
		\item $Y$ is right tail decreasing in $X$ (RTD($Y\mid X$)), if $P[Y > y \mid X > x]$ is a nonincreasing function of $x$ for all $y$.
		\item $X$ is right tail decreasing in $Y$ (RTD($X\mid Y$)), if $P[X >x \mid Y > y]$ is a nonincreasing function of $y$ for all $x$. 
		\item $Y$ is stochastically decreasing in $X$ denoted as SD($Y\mid X$), (also known as negatively regression dependent ($Y\mid X$)) if
		$P[Y > y \mid X = x]$ is a nonincreasing function of $x$ for all $y$. 
		\item  $X$ is stochastically decreasing in $Y$ denoted as SD($X\mid Y$), (also known as negatively regression dependent ($X\mid Y$)) if 
		$P[X > x\mid  Y = y]$ is a nonincreasing function of $y$ for all $x$.
		\item \label{eqn_NLR} Let $X$ and $Y$ be continuous random variables with joint density function $h(x,y)$. Then $X$ and $Y$ are negatively likelihood ratio
		dependent, denote by NLR(X,Y), if 
		$h(x_1,y_1)h(x_2,y_2)\leq h(x_1,y_2)h(x_2,y_1)$
		for all $x_1, x_2, y_1, y_2\in I$ such that $x_1\leq x_2$ and $y_1\leq y_2$.
	\end{enumerate}
\end{definition}
Now, in the following theorems, we establish that the proposed copula $C_{\theta}$ satisfies all the aforementioned dependence properties.
The detailed proofs are provided in Appendix B. 

\begin{theorem}\label{thm6}
	Let $X$ and $Y$ be two random variables with copula $C_\theta$. Then 
	$(i)$ $X$ and $Y$ are LTI($Y\mid X$),
	$(ii)$ $X$ and $Y$ are LTI($X\mid Y$),
	$(iii)$ $X$ and $Y$ are RTD($Y\mid X$), and
	$(iv)$ $X$ and $Y$ are RTD($X\mid Y$).
\end{theorem}

\begin{theorem}\label{thm7}
	Let $X$ and $Y$ be two random variables with copula $C_\theta$. Then $(i)$ $X$ and $Y$ are SD($Y\mid X$), and
	$(ii)$ $X$ and $Y$ are SD($X\mid Y$).
\end{theorem}

\begin{theorem}\label{NLR_thm}
	Let $X$ and $Y$ be two random variables with copula $C_\theta$. Then $X$ and $Y$ are NLR.
\end{theorem}

\begin{remark}
	{\rm 
		Two random variables $X$ and $Y$ with copula $C_{\theta}$  are NQD. This directly follows from Theorem \ref{NLR_thm}. See the interrelationships between different concepts of negative dependence summarised in \citep[p-130]{Dist} for details.}
\end{remark}
\section{Ordering Properties}\label{OrderP}
In Section \ref{NProp}, several negative dependence properties of the proposed copula $C_\theta$ has been investigated for the fixed $\theta>0$. In this section, we discuss the ordering properties of the proposed copula $C_\theta$, which provides a precise (and also intuitively expected) notion for one bivariate distribution being more positively or negatively associated than another. For this purpose, we first recall the definitions of the dependence orderings for bivariate distributions.
These definitions describe the strength of dependence of a copula with respect to its dependence parameter $\theta$.  
\cite{Leh} was first to introduce the NQD and NRD notions. Following this
notions, \cite{Yanagimoto_Okamoto_1969} introduced the ordering properties as defined below.
\begin{definition}
	Let $F$ and $G$ be two bivariate distributions with the same marginals. Then $F$ is said to be smaller than $G$ in the NQD sense denoted as $F\prec_{NQD}G$ if $$F(x,y)\geq G(x,y)\,\,\,\,\,\,\forall x\,\, {\rm and}\,\, y.$$
\end{definition}

\begin{definition}\label{def42}
	Let $F$ and $G$ be two bivariate distributions with the same marginals, and let $(U, V)$ and $(X, Y)$ be two random vectors having the
	distributions $F$ and $G$, respectively. Then $F$ is said to be smaller than
	$G$ in the NRD sense, denoted by $F\prec_{NRD} G $ or $(U, V)\prec_{NRD} (X, Y)$ if, for any $x_1\leq x_2$,
	\begin{equation*}
		F^{-1}_{V\mid U}(u\mid x)\geq F^{-1}_{V\mid U}(v\mid x^{'})\implies G^{-1}_{V\mid U}(u\mid x)\geq G^{-1}_{V\mid U}(v\mid x^{'})
	\end{equation*}
	for any $u, v \in I$, where $F_{V \mid U}$ denote the conditional distribution of $V$ given $U = u$ and  $F^{-1}_{V\mid U}$ denote its right-continuous inverse. Equivalently, $F\prec_{NRD} G $ if and only if $G^{-1}_{Y\mid X}\left[F_{V\mid U}(y\mid x)\mid x\right]$ is decreasing in $x$ for all $y$ \citep{Fang_Joe_1992}. 
\end{definition}
Later, \cite{Kimeldorf_Sampson_1987} have introduced and studied in detail the notion of the Negatively Likelihood Ratio dependence ordering that is described in the following definition.
Let the random variables $X$ and $Y$ have the joint distribution $G(x,y)$. For any two intervals $I_1$ and $I_2$ of the real line, let us denote $I_1\leq I_2$ if $x_1\in I_1$ and $x_2\in I_2$ imply that $x_1\leq x_2$. For any two intervals $I$ and $J$ of the real line let $G(I, J)$ represent the probability assigned by $G$ to the rectangle $I\times J$.
\begin{definition}\label{def_NLR}
	Let $F$ and $G$ be two bivariate distributions with the same marginals, and let $(U, V)$ and $(X, Y)$ be two random vectors having the distributions $F$ and $G$, respectively.  Then $F$ is said to be smaller than $G$ in the NLR dependence sense, denoted by $F \prec_{NLR} G$ or $(U, V) \prec_{NLR} (X, Y)$ if
	$ F(I_1, J_1)F(I_2, J_2) G(I_1, J_2)G(I_2, J_1)\geq F(I_1, J_2)F(I_2, J_1) G(I_1, J_1)G(I_2, J_1)$
	whenever  $I_1\leq I_2$ and $J_1\leq J_2$. When the densities $F$ and $G$ exist and denoted by $f$ and $g$, respectively, then the aforementioned condition equivalently is written as 
	$f(x_1, y_1)f(x_2, y_2) g(x_1, y_2)g(x_2, y_1)\geq f(x_1, y_2)f(x_2, y_1) g(x_1, y_1)g(x_2, y_1) $
	whenever  $x_1\leq x_2$ and $y_1\leq y_2$.
	
	
\end{definition}

In the following theorems, we derive the sufficient conditions under which one bivariate distribution will be more negatively associated than another. The detailed proofs of the following theorems are presented in Appendix C.

\begin{theorem}\label{thm9}
	If $\theta_1\leq \theta_2$, then $C_{\theta_1}(u,v)\prec_{NQD}C_{\theta_2}(u,v).$
\end{theorem}

\begin{theorem}\label{thm10}
	If $\theta_1\leq \theta_2$, then $C_{\theta_1}(u,v)\prec_{NRD}C_{\theta_2}(u,v).$
\end{theorem}

\begin{theorem}\label{thm11}
	If $\theta_1\leq \theta_2$, then $C_{\theta_1}(u,v)\prec_{NLR}C_{\theta_2}(u,v).$
\end{theorem}

\section{Examples}\label{Exam}
Traditionally, bivariate life distributions available in the literature are positively correlated \citep{Dist}. However, in many real life scenarios, paired observations of non-negative variables are negatively correlated \citep{Bhuyan_2020}. For example, the rainfall intensity and the duration are jointly modeled incorporating their negative dependence for the study of the corresponding flood frequency distribution \citep{Rain}. \citet{Gum} and \citet{Fre} have proposed the bivariate Exponential distributions with lower bound of the correlation coefficient as $-0.4$. 
In this section, several specific families of bivariate distributions are generated using the proposed copula (\ref{Copula}) with different choices for marginal distribution. For modelling purposes, the Lognormal, Weibull, and Gamma distributions are popular among practitioners in the fields of engineering, medical science, and environmental science \citep{OZ_2016,Wind_2017, Brain_2019}. We consider these choices as baseline distribution. We first define a bivariate Weibull and bivariate Gamma distribution. Then we consider a case when the marginals are different, one from the Lognormal and another from the Weibull family. It should be noted that the resulting bivariate distributions can be described implementing all notions of negative dependence discussed in Section \ref{NProp} and \ref{OrderP}. 

\begin{example}\label{bvweibull}
	{\bf Bivariate Weibull distribution:}
	{\rm 
		A family of bivariate Weibull distributions based on the proposed copula $C_{\theta}$, with marginals $F(x) =\left[ 1 - e^{-(\lambda_1 x)^{\delta_1}}\right]\mathds{1}(x>0)$, and $G(y) = \left[1 - e^{-(\lambda_2 y)^{\delta_2}}\right]\mathds{1}(y>0)$, is given by
		\begin{equation*}\label{BW}
			h(x,y) = \begin{dcases}
				\dfrac{\delta_1\delta_2\lambda_1^{\delta_1}\lambda_2^{\delta_2}\theta^{\theta+1}}{(1+\theta)^{\theta}}x^{\delta_1 -1}y^{\delta_2 - 1} \left(\dfrac{e^{-(\lambda_1 x)^{\delta_1}}}{1-e^{-(\lambda_2 y)^{\delta_2}}}\right)^{1+\theta}, &
				0<y \leq\phi_1, x>\phi_2(y) \\
				\delta_1\delta_2\lambda_1^{\delta_1}\lambda_2^{\delta_2}(1+\theta)x^{\delta_1 -1}y^{\delta_2 - 1}e^{-(\lambda_2 y)^{\delta_2}}\left(e^{-(\lambda_1 x)^{\delta_1}}\right)^{1+\theta}, &  x>0, y>\phi_1
			\end{dcases}
		\end{equation*}
		where $\phi_1 =\dfrac{1}{\lambda_2}\left[\log(1+\theta)\right]^{\frac{1}{\delta_2}}$, $\phi_2(y)= \dfrac{1}{\lambda_1}\left[\log\left(\dfrac{\theta}{(1+\theta)(1 - e^{-(\lambda_2 y)^{\delta_2}})}\right)\right]^{\frac{1}{\delta_1}}$, $\lambda_{i}>0$, $\delta_{i}>0$ for $i=1,2$.

}\end{example}

\begin{example} {\bf Bivariate Gamma distribution:}\label{bvgamma} {\rm A family of bivariate Gamma distributions based on the proposed copula $C_{\theta}$, with marginals $F(x)=\left[\int_0^x\frac{1}{\Gamma(\alpha_1)}\beta_1^{\alpha_1}x^{\alpha_1 -1}e^{-\beta_1 x} \right]\mathds{1}(x>0)$, and $G(y) =\left[\int_0^y\frac{1}{\Gamma(\alpha_2)}\beta_2^{\alpha_2}y^{\alpha_2 -1}e^{-\beta_2 y} \right]{\mathds{1}(y>0)}$, is given by
		
		\begin{equation*}
			h(x,y) = \begin{dcases}
				\dfrac{\beta_1^{\alpha_1}\beta_2^{\alpha_2}\theta^{1+\theta}x^{\alpha_1 -1}y^{\alpha_2 -1}e^{-(\beta_1 x+\beta_2 y)}}{\Gamma(\alpha_1)\Gamma(\alpha_2)(1+\theta)^\theta}\left[1-\dfrac{\gamma_1(\alpha_1, \beta_1 x)}{\Gamma(\alpha_1)}\right]^\theta \left[\dfrac{\gamma_2(\alpha_2, \beta_2 y)}{\Gamma(\alpha_2)}\right]^{-(1+\theta)}, \\\hspace{200pt}0<y \leq\xi_2 , \xi_1(y)<x< \eta \\
				\dfrac{\beta_1^{\alpha_1}\beta_2^{\alpha_2}(1+\theta)}{\Gamma(\alpha_1)\Gamma(\alpha_2)}x^{\alpha_1 -1}y^{\alpha_2 -1}e^{-(\beta_1 x+\beta_2 y)}\left[1-\dfrac{\gamma_1(\alpha_1, \beta_1 x)}{\Gamma(\alpha_1)}\right]^\theta,\\\hspace{200pt}
				0<x<\eta, 
				\zeta_1<y<\zeta_2,
			\end{dcases}
		\end{equation*}
		where $\zeta_1=\gamma_2^{-1}\left(\dfrac{\theta}{1+\theta}\right)$, $\zeta_2 = \gamma_2^{-1}(\Gamma(\alpha_2))$, $\xi_2= \gamma_2^{-1}\left(\dfrac{\Gamma(\alpha_2)\theta}{1+\theta}\right)$, $\eta = \gamma_1^{-1}(\Gamma(\alpha_1))$, \\
		$\xi_1(y)= \gamma_1^{-1}\left[\Gamma(\alpha_1) \left(1-\dfrac{(1+\theta)\gamma_2(\alpha_2,\beta_2 y)}{\theta \Gamma(\alpha_2)}\right)\right]$, $\gamma_i(\alpha_i, \beta_i) = \int_0^{\beta_i} t^{\alpha_i - 1} e^{-t}dt$, $\alpha_{i}>0$, $\beta_{i}>0$ for ${i= 1,2}$.
	}
\end{example}

\begin{example} {\bf Bivariate Lognormal-Weibull distribution:}\label{bvlnw}
	A family of bivariate distribution with one marginal from Lognormal distribution and another from Weibull distribution based on the proposed copula $C_{\theta}$, with marginal distribution functions $F(x) =\frac{1}{2}\left[ 1 + erf\left(\frac{\ln{x}-\mu}{\sqrt{2}\sigma}\right)\right]\mathds{1}(x>0)$, and $G(y) = \left[1 - e^{-(\lambda y)^{\delta}}\right]{\mathds{1}(y>0)}$, is given by
	\begin{equation*}\label{BW1}
		h(x,y) = \begin{dcases}
			\dfrac{\delta\lambda^\delta\theta^{\theta+1}}{\sigma\sqrt{\pi} (1+\theta)^{\theta} 2^{\frac{2\theta +1}{2}}}  \dfrac{y^{\delta -1}e^{-(\lambda y)^\delta}}{x\left(1-e^{-(\lambda y)^\delta}\right)^{1+\theta}}\left[ 1 - erf\left(\frac{\ln{x}-\mu}{\sqrt{2}\sigma}\right)\right]^\theta, \\\hspace{190pt}
			0<y \leq\psi_1, \psi_2(y)<x<\psi_3 \\
			\dfrac{\delta(1+\theta)\lambda^{\delta}}{\sigma\sqrt{\pi}2^{\frac{2\theta +1}{2}} } \frac{y^{\delta -1}e^{-(\lambda y)^\delta}}{x}\left[ 1 - erf\left(\frac{\ln{x}-\mu}{\sqrt{2}\sigma}\right)\right]^\theta , \\\hspace{210pt} x>0, y>\psi_1
		\end{dcases}
	\end{equation*}
	where $\psi_1 =\frac{1}{\lambda}\left[\log(1+\theta)\right]^{\frac{1}{\delta}}$, $\psi_2(y)=\exp\left[\mu+\sigma \sqrt{2}\,erf^{-1}\left\lbrace 1-\frac{2(1+\theta)}{\theta}\left(1-e^{-(\lambda y)^\delta} \right)\right\rbrace\right]$, $\psi_3=\mu+\sigma \sqrt{2}\,erf^{-1}\left(\frac{1}{2}\right)$, $\lambda>0$, $\delta>0$, $-\infty<\mu<\infty$, $\sigma>0$, and $erf(x)=\frac{2}{\sqrt{\pi}}\int_0^x e^{-t^2}dt$. 
\end{example}

\begin{remark}
	The bivariate Weibull (in Example \ref{bvweibull}) and the bivariate Gamma (in Example \ref{bvgamma}) reduce to bivariate Exponential distribution for $\delta_1=\delta_2 = 1$, and $\alpha_1=\alpha_2 = 1$, respectively.
\end{remark}

\section{Estimation Methodology}\label{EST}
In a classical parametric setting, a straightforward approach is to estimate the dependence parameter and the parameters associated with the marginals using maximum likelihood method. This method is theoretically valid but there are some practical limitations. Firstly, the estimation of the dependence parameter $\theta$ depends on the parametric assumptions made on the marginals and the estimate of $\theta$ will be biased if the marginals are misspecified. The second drawback is computational as the log-likelihood function involves potentially large number of parameters and high-dimensional optimization is known to be challenging. See \citet[Ch-4]{Mariu_2018} for details. To avoid aforementioned computational burden \citet{Joe_1997} proposed a two stage method known as inference function for margins (IFM). This estimation method is based on two separate maximum likelihood estimations
of the univariate marginal distributions, followed by an optimization of the bivariate likelihood as a function of the dependence parameter. Similar to maximum likelihood estimate, the estimate of $\theta$ based on IFM may be biased if the margins are partially misspecified \citep[p-136]{Mariu_2018}. Although the IFM has computational edge, it is less efficient compared to the maximum likelihood estimate \citep[Ch-5]{Joe_2006}.

We propose to use a method that
close in spirit to the method of inference function for margins (IFM) but avoids the issue with misspecified marginals for the estimation of $\theta$. In contrast to IFM, we do not maximize the bivariate likelihood. Instead, we determine the dependence parameter using method of moments \citep[p-141]{Mariu_2018}. The method of fitting a bivariate distribution with marginals $F_{\eta_{i}}(\cdot)$, indexed by parameter $\eta_{i}$ for $i=1,2$, involves the following steps:
\begin{itemize}
	\item[(i)] Obtain the estimates $\hat{\eta}_{i}$ for $i=1,2$ using maximum likelihood method.
	\item[(ii)] Estimate of $\theta$ is given by $\hat{\theta}=\frac{-\tau_{n}}{1+\tau_{n}}$, or obtained by solving $\rho_{n} =\dfrac{2(3+3\hat{\theta}+\hat{\theta}^2)}{2+3\hat{\theta}+\hat{\theta}^2}-3$, where $\tau_{n}$ and $\rho_{n}$ are sample version of Kendall's $\tau$ and Spearman's $\rho$, respectively.
	\item[(iii)] Obtain the fitted bivaritae distribution by putting $F_{\hat{\eta}_{1}}(\cdot)$, and $G_{\hat{\eta}_{2}}(\cdot)$, and $\hat{\theta}$ in (\ref{Copula}).
\end{itemize}
These steps are easy to execute and familiar to the practitioners of different fields of science. This method allows the copula to adequately approximate the dependence structure of the bivariate data, which is of prime concern from a practical point of view.

\section{Application}\label{Case}
\subsection{Exploratory Data Analysis}\label{EDA}
For an illustrative data analysis based on the proposed copula, we consider a data set on daily air quality measurements for 153 days in the New York Metropolitan Area from May 1, 1973, to September 30, 1973. Information on average wind speed (in miles per hour) and mean ozone level (in parts per billion), were obtained from the New York State Department of Conservation and the National Weather Service, USA. This data set is openly available in R software. See \citet[Ch 2-5]{Air} for the detailed description of the data. Ozone in the upper atmosphere protects the earth from the sun's harmful rays. On the contrary, exposure to ozone also can be hazardous to both humans and some plants in the lower atmosphere. Variations in weather conditions play an important role in determining ozone levels \citep{Metero1, Metero2}. In general, concentration of the ozone level is affected by a wind speed. High winds tend to disperse pollutants, which in turn, dilute the concentration of the ozone level. However, stagnant conditions or light winds allow pollution levels to build up and thereby, the ozone level too becomes larger. Environmental scientists and meteorologists are interested in the study of the effect of a wind speed on the distribution patterns of ozone \citep{Wind} levels. For our analysis, we consider 116 observations discarding the missing values and presented the scatter plot of average wind speed versus ozone levels in Figure \ref{Scatter}. It indicates strong negative dependence, and we find that Spearman's rho and Kendall's tau coefficients are -0.59 and -0.43, respectively. Further, we apply the methodology proposed by \citet{Lu} based on Kolmogorov–Smirnov (KS), Anderson–Darling (AD), and Cramér-vonMises (CvM) discrepancy measures to test the hypothesis if the true underlying copula satisfies the NQD property. The p-values corresponding to KS, AD and CvM tests are 0.893, 0.571, and 0.861, respectively, affirm a strong notion of negative dependence between average wind speed and ozone levels in the NQD sense.

\subsection{Modeling Wind Speed and Ozone Level}\label{Modeling}
In the field of engineering and environmental science, Lognormal, Weibull, and Gamma distributions are widely used for modeling  wind speed recorded in the same location \citep{Wind_1978,Wind_2015,Wind_2017,Wind_2020,Wind_2021}. These distributions are also used for modeling the level of various pollutants and ozone level \citep{OZ_2016,OZ_2018,OZ_2021}. Therefore, we consider these three models for estimation of the parameters associated with the marginal distributions of the wind speed and the mean ozone level. Based on the Akaike information criterion, the Gamma distribution fits both marginals better as compared with other choices. The maximum likelihood estimates of the shape and the scale parameters are obtained as 7.171 and 1.375, respectively, for the wind speed, and the same for the mean ozone levels are 1.7 and 24.775, respectively. The estimate of the dependence parameter is obtained as $\hat{\theta}=0.765$. Therefore, the joint distribution of wind speed and the mean ozone level is represented by the bivariate Gamma distribution provided in Example \ref{bvgamma}, and presented graphically in Figure \ref{CDF_Data}. Following \cite{Bala_2016}, we then use bootstrap based Kolmogorov-Smirnov test to check whether the Gamma distribution is a good fit for the marginals. Also, we evaluate the goodness of fit of the proposed copula based on Kolmogorov-Smirnov statistic utilising the bootstrap algorithm proposed by \citet{Genest_2006}. We find the proposed model fits the data reasonably well. The R programme related to the proposed estimation methodology are provided in the Supplementary material. In Figure \ref{Contour} we present the contour plot of the the distribution of wind speed and mean ozone level. It indicates that the concentration of mean ozone level varies from 6-30 ppb when the wind speed is within 7-16 mph. The estimated conditional distributions of the mean ozone level keeping the wind speed fixed at the empirical first decile (5.7 mph), median (9.7 mph), and ninth decile (14.9 mph) are presented in Figure \ref{CND_CDF}. It is easy to see that the distribution of the mean ozone level decreases stochastically (in the sense of the usual stochastic order) as the wind speed increases. This visual representation of the regression dependence property indicates that the ozone level distributions below the level of 60 ppb differ significantly with wind speed. This can assist in formulating policies and guidelines to choose between locations to avoid health hazards related to high ozone levels.

\begin{figure}[htp]
	\centering
	\subfigure [Scatter plot of wind speed versus ozone.] {
		\includegraphics[width=70mm,height=70mm]{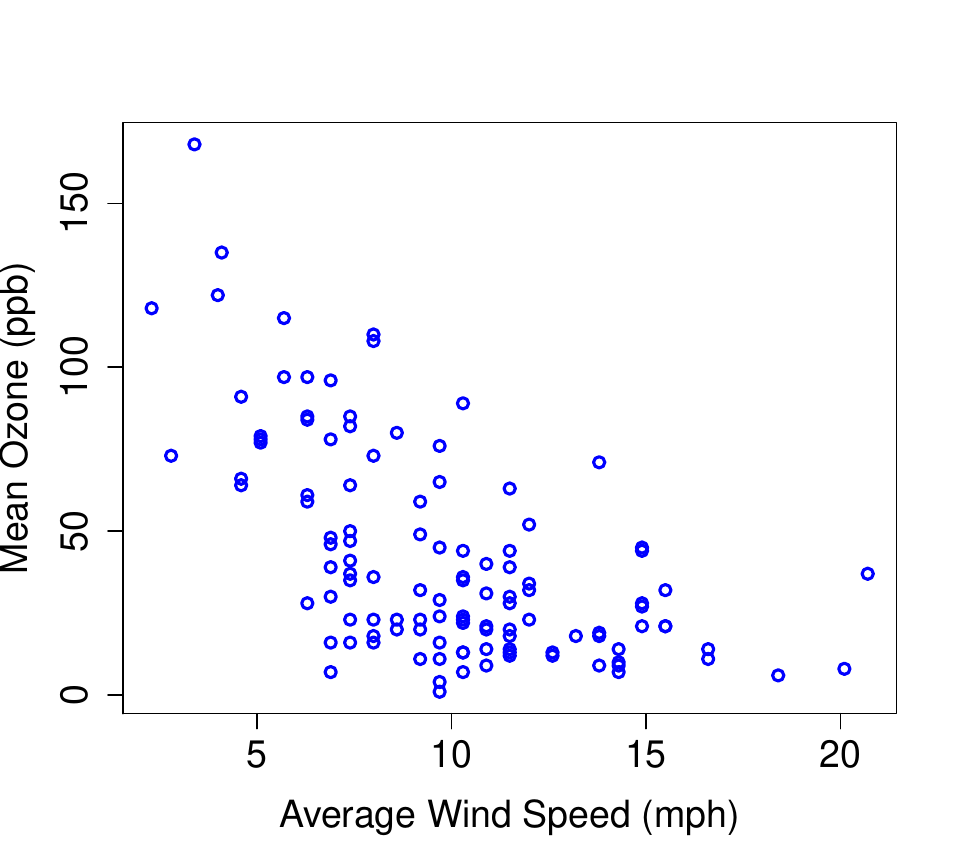}
		\label{Scatter}
	}
	\subfigure[Distribution of wind speed and ozone.]{
		\includegraphics[width=70mm,height=70mm]{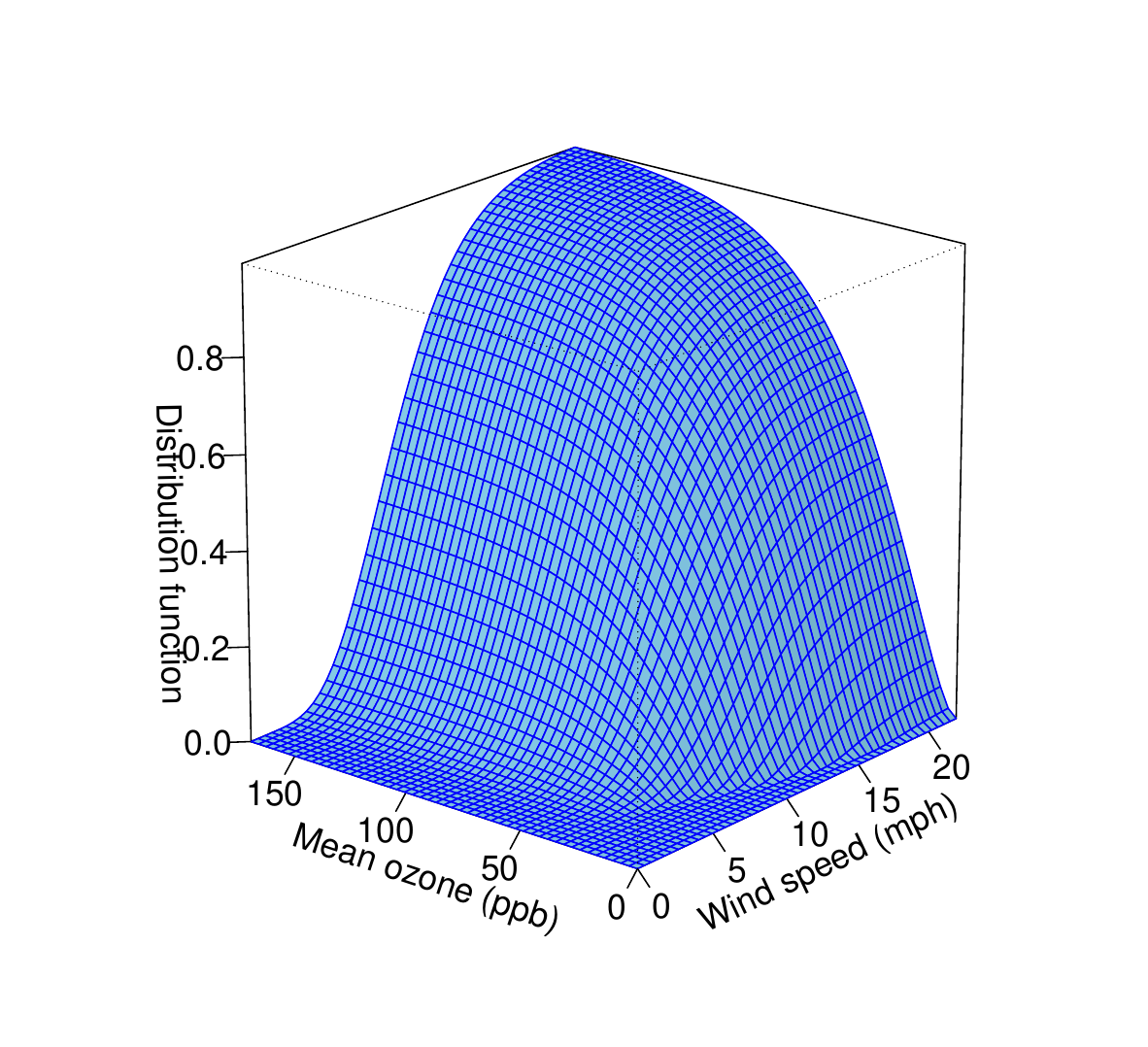}
		\label{CDF_Data}
	}
	\subfigure[Contour plot of wind speed and ozone.]{
		\includegraphics[width=70mm,height=70mm]{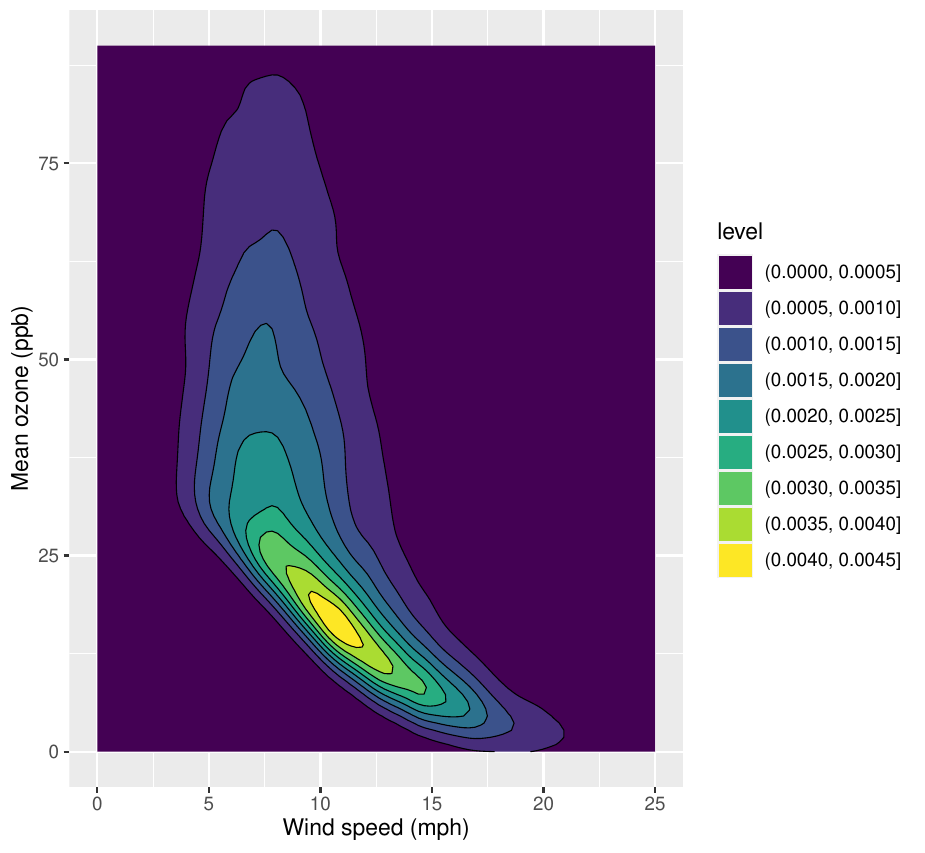}
		\label{Contour}	
	}
	\subfigure[Effect of wind speed on ozone.]{
		\includegraphics[width=70mm,height=70mm]{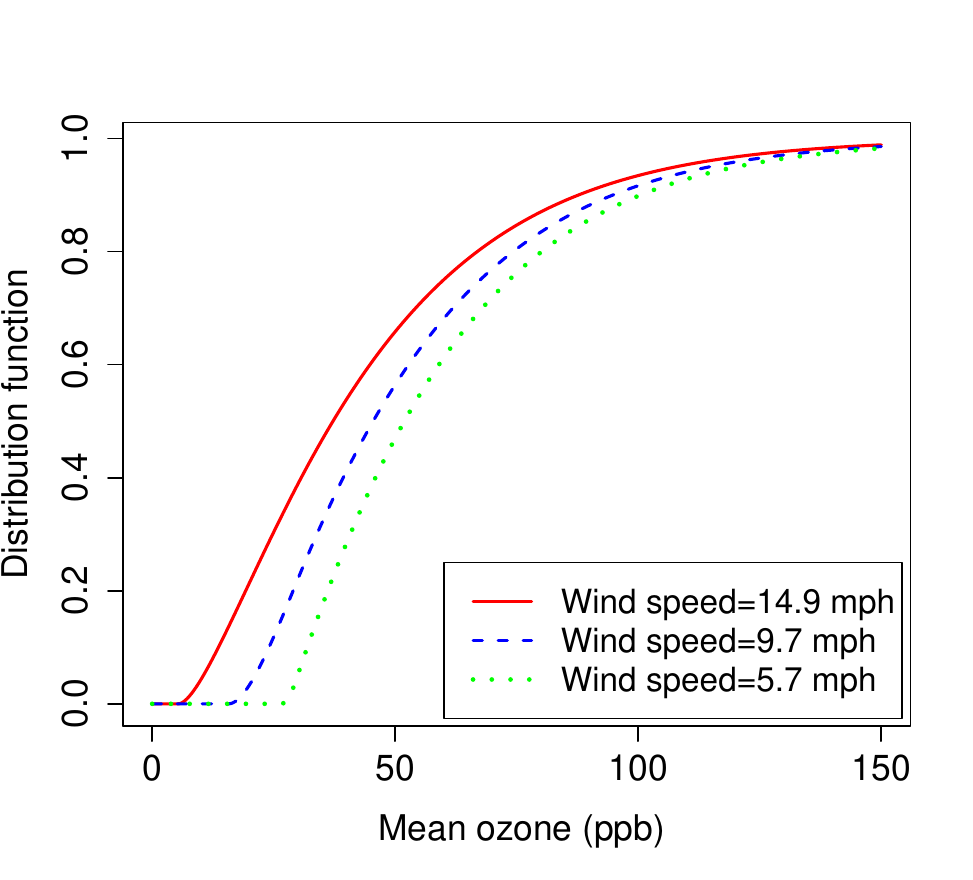}
		\label{CND_CDF}
	}
	\caption{Results based on the analysis of New York air quality data.}
\end{figure}

\section{Concluding Remarks}\label{Con}
We construct the new flexible bivariate copula for modeling negative dependence between two random variables. Its  correlation coefficient takes any value in the interval $(-1,0)$, which was not the case for other copulas reported in the literature. It is important to note that the Spearman's rho and the Kendall's tau have a simple one-parameter form with negative values in the full range. The properties of the proposed copula is an agreement with  most of the popular notions of negative dependence available in the literature, namely quadrant Dependence, regression dependence and likelihood ratio dependence, etc. It is an interesting problem to consider a semi-parametric generalisation of the proposed copula and investigate its associated properties. Another possible direction of future research could be a multivariate extension of the proposed copula using the approaches considered by \cite{Fischer_Kock_2008} and \citep{Mazo_et_al_2015}.

For an illustrative data analysis based on the proposed copula, we consider a data set on daily air quality measurements for New York Metropolitan Area. Based on the observed data, we find that wind speed and ozone levels strongly dependent in the NQD sense. We consider three different models (Lognormal, Weibull, and Gamma distributions) for estimation of  parameters associated with the marginal distributions of the wind speed and the  mean ozone level. It is shown that the Gamma distributions fits better for both marginals and that the distribution of the mean ozone level decreases stochastically (in the sense of the usual stochastic order) as the wind speed increases. The scope of the proposed copula goes far beyond this particular application. For example, biomedical researchers can utilize the proposed copula in studying the negative association between BMI and glycated proteins \citep{BMI}. One can also extend the proposed copula to asses and model the nonlinear and asymmetric negative dependence over time in security and commodity markets \citep{Market}.




	\section*{Appendix A }\label{secA}

	{\bf A1. Proof of Proposition \ref{pro21}.}
	
	{\bf Case I.} For $0<v \leq\frac{\theta}{1+\theta}$, and $1-\frac{(1+\theta)v}{\theta}<u<1$, we have
	\begin{eqnarray*}
		\dfrac{\partial C_\theta}{\partial \theta} &=& \dfrac{\theta^\theta}{(1+\theta)^{(1+\theta)}}(1-u)^{(1+\theta)}v^{-\theta}\left[ \log\left(\dfrac{\theta}{1+\theta}\right) +\log(1-u)-\log(v) \right]\\
		&\leq& \dfrac{\theta^\theta}{(1+\theta)^{(1+\theta)}}(1-u)^{(1+\theta)}v^{-\theta}\left[ \log\left(\dfrac{\theta}{1+\theta}\right) +\log\left[\dfrac{(1+\theta)v}{\theta}\right]-\log(v) \right],\\
		&&\hspace{200pt}\text{since $(1-u)\leq \dfrac{(1+\theta)v}{\theta}$}\\
		&=&0
	\end{eqnarray*}

	{\bf Case II.} For $0<u<1$, and $\frac{\theta}{1+\theta}<v<1$, we have
	$$\dfrac{\partial C_\theta}{\partial \theta}= (1-u)^{(1+\theta)}(1-v)\log(1-u)\leq 0.$$
	
	Now combining Case I and II, we have $\dfrac{\partial C_\theta}{\partial \theta}\leq 0$ for all $(u,v)\in I^2$, which implies $C_{\theta}$ is decreasing in $\theta$.
	\vspace{5pt}
	\\
	{\bf A2. Proof of Proposition \ref{pro22}.}
	
	{\bf Case I.} For $0<v \leq\frac{\theta}{1+\theta}$, and $1-\frac{(1+\theta)v}{\theta}<u<1$, we have 
	\begin{eqnarray*}
		\nabla^2 C_\theta(u,v)& =& \dfrac{\partial^2 C_\theta(u,v)}{\partial u^2} + \dfrac{\partial^2 C_\theta(u,v)}{\partial v^2}\\
		&= & \dfrac{\theta^{(1+\theta)}}{(1+\theta)^\theta}\left[(1 - u)^{(\theta-1)} v^{-\theta} + (1-u)^{(1+\theta)}v^{-(2+\theta)}\right]\geq 0
	\end{eqnarray*}
	{\bf Case II.} For $0<u<1$, and $\frac{\theta}{1+\theta}<v<1$, we have
	\begin{eqnarray*}
		\nabla^2 C_\theta(u,v)& =& \dfrac{\partial^2 C_\theta(u,v)}{\partial u^2} + \dfrac{\partial^2 C_\theta(u,v)}{\partial v^2}\\
		&= & \theta (1 + \theta) (1 - u)^{(\theta-1)} (1 - v)\geq 0
	\end{eqnarray*}
	Now from Case I and II we can write $\nabla^2 C_\theta(u,v) \geq 0$ for all $(u,v)\in I^2$, and hence the result follows.
	\vspace{5pt}
	\\
	{\bf A3. Proof of Proposition \ref{pro23}.}
	
	To establish the absolute continuity of the proposed copula $C_{\theta}$, it is required to show  $$\int_0^u\int_0^v\dfrac{\partial^2}{\partial s\partial t}C_\theta(s,t)dtds =C_\theta(u,v),$$
	for every $(u,v)\in I^2$.\\
	\\
	{\bf Case I.} For $0<v \leq\frac{\theta}{1+\theta}$, and $1-\frac{(1+\theta)v}{\theta}<u<1$, we have
	\begin{eqnarray*}
		\int_0^u\int_0^v\dfrac{\partial^2}{\partial s\partial t}C_\theta(s,t)dtds &=& \int_{1-\frac{(1+\theta)v}{\theta}}^u\int_{\frac{\theta(1-s)}{(1+\theta)}}^v\dfrac{\theta^{1+\theta}}{(1+\theta)^\theta}(1-s)^\theta t^{-(1+\theta)}dtds\\
		&=& \int_{1-\frac{(1+\theta)v}{\theta}}^u\left[1-\left(\frac{\theta}{1+\theta}\right)^\theta(1-s)^\theta v^{-\theta}\right] ds\\
		&=&\int_{1-u}^{\frac{(1+\theta)v}{\theta}}\left[1-\left(\frac{\theta}{1+\theta}\right)^\theta z^\theta v^{-\theta}\right] dz \,\,\,({\rm where}\,\, z = 1-s)\\
		&=& v-(1-u)+\dfrac{\theta^\theta}{(1+\theta)^{1+\theta}}(1-u)^{1+\theta}v^{-\theta}=C_\theta(u,v). 
	\end{eqnarray*}  
	{\bf Case II.} For $0<u<1$, and $\frac{\theta}{1+\theta}<v<1$, we have    
	\begin{eqnarray*}
		\int_0^u\int_0^v\dfrac{\partial^2}{\partial s\partial t}C_\theta(s,t)dtds &=& \int_0^u\int_{\frac{\theta(1-s)}{(1+\theta)}}^{\frac{\theta}{(1+\theta)}}\dfrac{\theta^{1+\theta}}{(1+\theta)^\theta}(1-s)^\theta t^{-(1+\theta)}dtds\\&+&\int_0^u\int_{\frac{\theta}{(1+\theta)}}^v(1+\theta)(1-s)^\theta dtds\\
		&=& \int_0^u\left[1-(1-s)^\theta\right] ds + \int_0^u \left[ v-\theta(1 - v)\right](1 - s)^\theta\\
		&=& u-\dfrac{1}{1+\theta}+\dfrac{(1-u)^{\theta+1}}{\theta + 1}+ \left[ v-\theta(1 - v)\right]\dfrac{\left[1-(1-u)^{(\theta+1)}\right]}{1+\theta}\\
		&=& u-(1-v)\left[1-(1-u)^{(\theta+1)}\right]=C_\theta(u,v). 
	\end{eqnarray*} 
	Therefore, the results follows by combining Case I and II.

	
	
	\section*{Appendix B}\label{secB}%

	{\bf B1. Proof of Theorem \ref{thm6}.}
	
	(i) To establish LTI($Y\mid X$), it is sufficient to show that for any $v$ in $I$, $\frac{C(u,v)}{u}$ is nondecreasing in $u$ \citep[Theorem 5.2.5, p-192]{Nelson_2006}.
	For $0<u<1$, and $\frac{\theta}{1+\theta}<v<1$, we have $$\dfrac{\partial}{\partial u}\left[\dfrac{C(u,v)}{u}\right]=\dfrac{(1-v)[1-(1-u)^\theta(1+\theta u)]}{u^2}.$$ Now we need to prove that $[1-(1-u)^\theta(1+\theta u)]>0$. Define $h(u):=(1-u)^\theta(1+\theta u)$. Observe that $h(0)=1$, $h(1)=0$, and $h(u)$ is a decreasing function in $u$, since  $h^{'}(u)=-\theta^2(1+\theta)u(1-u)^{(\theta-1)}<0$ for all  $u\in(0,1)$. Therefore, $\frac{\partial}{\partial u}\left[\frac{C(u,v)}{u}\right]>0$.\\ 
	
	Similarly, 
	for $0<v \leq\frac{\theta}{1+\theta}$, and $1-\frac{(1+\theta)v}{\theta}<u<1$, it can be shown that
	$$\dfrac{\partial}{\partial u}\left[\dfrac{C(u,v)}{u}\right]=- \dfrac{\frac{\theta^\theta}{(1+\theta)^{(1+\theta)}}(1+\theta u)(1-u)^\theta +v^{(1+\theta)}-v^\theta}{ u^2v^\theta}>0.$$
	Hence, the result follows.\\
	
	(ii) In view of Theorem 5.2.5 in \cite[p-192]{Nelson_2006}, the necessary and sufficient condition for LTI($X\mid Y$) is that, $\frac{C(u,v)}{v}$ is nondecreasing in $v$, for any $u$ in $I$.
	
	For $0<v \leq\frac{\theta}{1+\theta}$, and $1-\frac{(1+\theta)v}{\theta}<u<1$, we have
	$$\dfrac{\partial}{\partial v}\left[\dfrac{C(u,v)}{v}\right]= \dfrac{(u-1)\left[\frac{\theta^\theta}{(1+\theta)^{\theta}}(1-u)^\theta- v^{\theta}\right]}{v^{\theta+2}}\geq 0,$$ since $(u-1)<0$ and $\left[\frac{\theta^\theta}{(1+\theta)^{\theta}}(1-u)^\theta- v^{\theta}\right]<0$. \\
	
	Similarly, for $0<u<1$, and $\frac{\theta}{1+\theta}<v<1$, we have $$\dfrac{\partial}{\partial v}\left[\dfrac{C(u,v)}{v}\right]=\dfrac{(u-1)[(1-u)^\theta-1]}{v^2}\ge 0.$$
	Hence, the result follows.\\
	
	(iii) 
	To establish RTD($Y\mid X$), it is sufficient to show that $\frac{v-C(u,v)}{(1-u)}$ is a nondecreasing function in $u$ for any $v\in I$ \citep[Theorem 5.2.5, p-192]{Nelson_2006}.\\
	
	For $0<v \leq\frac{\theta}{1+\theta}$ and ${1-\frac{(1+\theta)v}{\theta}<u<1}$,  we have $$\dfrac{\partial}{\partial u}\left[\dfrac{v-C(u,v)}{(1-u)}\right]= \left(\dfrac{\theta}{1+\theta}\right)^{1+\theta}(1-u)^{\theta-1}v^{-\theta}>0.$$\\
	
	Similarly, for $0<u<1$, and  $\dfrac{\theta}{1+\theta}<v<1$, we have $$\dfrac{\partial}{\partial u}\left[\dfrac{v-C(u,v)}{(1-u)}\right]= (1-v)(1-u)^{\theta -1}>0.$$  Hence, the conclusion follows.\\

	(iv) By Theorem 5.2.5 in \cite[p-192]{Nelson_2006}, RTD($Y\mid X$) holds, if $\frac{u-C(u,v)}{(1-v)}$ is a nondecreasing function in $v$ for any $u\in I$. 
	
	For $0<v \leq\frac{\theta}{1+\theta}$ and $1-\frac{(1+\theta)v}{\theta}<u<1$, we have $$\dfrac{\partial}{\partial v}\left[\dfrac{u-C(u,v)}{(1-v)}\right]=\dfrac{\frac{\theta^\theta}{(1+\theta)^{(1+\theta)}}(1-u)^{1+\theta}v^{-(1+\theta)}[\theta(1-v)-v]}{(1-v)^2},$$ which is non-negative, since $v <\frac{\theta}{1+\theta}$.\\
	
	Similarly, for any fixed $u\in I$, and  $\frac{\theta}{1+\theta}<v<1$, $\frac{u-C(u,v)}{(1-v)}=1-(1-u)^{1+\theta}$ is a constant function in $v$. Hence the results follows.
	\vspace{5pt}
	\\
	{\bf B2. Proof of Theorem \ref{thm7}.}

	To establish SD($Y\mid X$) property of the proposed copula $C_{\theta}$, we utilise the geometric interpretation of the stochastic monotonicity given in Corollary 5.2.11 of \cite[p-197]{Nelson_2006}. Therefore, it is sufficient to show that $C_\theta(u,v)$ is a convex function of $u$. Similarly, SD($X\mid Y$) can be established by showing $C_\theta(u,v)$ is a convex function of $v$.\\
	
	(i)  For $0<v \leq\frac{\theta}{1+\theta}$, and $1-\frac{(1+\theta)v}{\theta}<u<1$, we have $$\dfrac{\partial^2}{\partial u^2}C_\theta(u,v) = \dfrac{\theta^{(1+\theta)}}{(1+\theta)^\theta}(1-u)^{\theta-1}v^{-\theta}>0.$$
	
	For $0<u<1,$ and  $\frac{\theta}{1+\theta}<v<1$, we have 
	$$\dfrac{\partial^2}{\partial u^2}C_\theta(u,v) =\theta (1+\theta)(1-v)(1-u)^{(\theta -1)}> 0.$$ Hence $C_\theta(u,v)$ is a convex function of $u$.\\  
	
	(ii)  For $0<v \leq\frac{\theta}{1+\theta}$, and $1-\frac{(1+\theta)v}{\theta}<u<1$, we have $$\dfrac{\partial^2}{\partial v^2}C_\theta(u,v) ={ \dfrac{\theta^{(1+\theta)}}{(1+\theta)^\theta}(1-u)^{1+\theta}v^{-(2+\theta)}}>0.$$
	
	Note that, for any fixed $u\in I$, and $\frac{\theta}{1+\theta}<v<1$,  $\frac{\partial}{\partial v}C_\theta(u,v)$ is a constant function of $v$. Hence, the result follows.
	\vspace{5pt}
	\\     
	{\bf B3. Proof of Theorem \ref{NLR_thm}.}
	
	To established the NLR between $X$ and $Y$ with copula $C_\theta$, we need to show $c_\theta(u_1,v_1)c_\theta(u_2,v_2)\leq c_\theta(u_1,v_2)c_\theta(u_2,v_1)$ holds for all $u_1\leq u_2$, and $v_1\leq v_2$, where $c_\theta(u,v)$ is the copula density given in (\ref{eqn_den}). Note that for the proposed copula $C_{\theta}$, the aforementioned condition holds with equality for all $u_1\leq u_2$ and $v_1\leq v_2$ in $I$.

	\section*{Appendix C}\label{secC}%
	{\bf C1. Proof of Theorem \ref{thm9}.}
	
	The results directly follow from Proposition \ref{pro21}. 
	\vspace{5pt}
	\\
	{\bf C2. Proof of Theorem \ref{thm10}.}
	
	Let $\theta_1\leq \theta_2$. The conditional copula of $V$ given $U=u$ is given by
	\begin{equation*}
		C_{\theta_1}(v\mid u)=\begin{dcases}
			1- \dfrac{\theta_1^{\theta_1}}{(1+\theta_1)^{\theta_1}}(1-u)^{\theta_1} v^{-\theta_1},  & \dfrac{(1-u)\theta_1}{(1+\theta_1)}<v <\dfrac{\theta_1}{1+\theta_1} \\
			1-(1+\theta_1)(1-v)(1-u)^{\theta_1}, &  \dfrac{\theta_1}{1+\theta_1}<v<1.
		\end{dcases}
	\end{equation*}
	Then $C_{\theta_2}^{-1}(C_{\theta_1}(v\mid u)\mid u)$ is given by    
	\begin{equation*}
		C_{\theta_2}^{-1}(C_{\theta_1}(v\mid u)\mid u) =\begin{dcases}
			\dfrac{\theta_2(1+\theta_1)^{(\theta1/\theta_2)}}{(1+\theta_2)\theta_1^{(\theta1/\theta_2)}}(1-u)^{1-{(\theta1/\theta_2)}}v^{(\theta1/\theta_2)},  & 0<v\leq1-(1-u)^{\theta_2}
			\\
			1-\dfrac{1+\theta_1}{1+\theta_2}(1-v)(1-u)^{(\theta1-\theta_2)}, &  1-(1-u)^{\theta_2}<v< 1.
		\end{dcases}
	\end{equation*}
	Note that $ C_{\theta_2}^{-1}(C_{\theta_1}(v\mid u)\mid u)$ is a decreasing function in $u$ as $\theta_1\leq \theta_2$. Now, using Definition \ref{def42}, the result follows.  
	\vspace{5pt}
	\\
	{\bf C3. Proof of Theorem \ref{thm11}.}
	
	Let $\theta_1\leq \theta_2$. Now, it is easy to verify that the condition provided in Definition \ref{def_NLR} holds for any choice of $u_1, u_2, v_1, v_2$, where $u_1\leq u_2$, $v_1\leq v_2$.

\section*{Acknowledgement}
The first author sincerely acknowledges the financial support from the University of the Free State, South Africa.
The second author was supported in part by
the Lloyd’s Register Foundation programme on data-centric engineering at the Alan Turing Institute, UK.

\section*{Supplementary Information}
Supplementary material is openly available at \doi{10.13140/RG.2.2.31152.74247}.



\bibliographystyle{apalike}
\bibliography{Prajamitra_ref}

\end{document}